\documentclass[a4paper]{amsart}
\usepackage[foot]{amsaddr}
\usepackage[T1]{fontenc}
\usepackage{fullpage}
\usepackage{amsmath,amsthm,amssymb,amscd,amsfonts}
\usepackage[dvipsnames]{xcolor}
\definecolor{darkgreen}{rgb}{0,0.5,0}
\definecolor{darkblue}{rgb}{0,0,0.5}
\usepackage[colorlinks=true,urlcolor=darkblue,citecolor=darkgreen,linkcolor=darkblue]{hyperref}
\usepackage{doi}
\usepackage{graphicx}
\usepackage{tabularray}
\usepackage{adjustbox}
\usepackage{tikz}
\usepackage{makecell}
\usepackage[version=4]{mhchem}

\usetikzlibrary{graphs, quotes}

\usepackage{multirow}%
\usepackage{mathrsfs}%
\usepackage[title]{appendix}%
\usepackage{xcolor}%
\usepackage{textcomp}%
\usepackage{manyfoot}%
\usepackage{booktabs}%
\usepackage{algorithm}%
\usepackage{algorithmicx}%
\usepackage{algpseudocode}%
\usepackage{listings}%
\usepackage{hhline}%

\definecolor{reb}{RGB}{30,144,255}

\newcommand{\bbC}{\mathbb{C}}
\newcommand{\bbE}{\mathbb{E}}

\newcommand{\bbR}{\mathbb{R}}

\newcommand{\calG}{\mathcal{G}}
\newcommand{\calP}{\mathcal{P}}

\newcommand{\expecO}{\langle O \rangle}
\newcommand{\expecP}{\langle P \rangle}
\newcommand{\estO}{\widehat{\expecO}}

\newcommand{\eps}{\varepsilon}

\newcommand{\QWC}{\mathrm{QWC}}
\newcommand{\FC}{\mathrm{FC}}
\newcommand{\HF}{\mathrm{HF}}
\newcommand{\Ha}{\mathrm{Ha}}

\DeclareMathOperator{\Bias}{Bias}
\DeclareMathOperator{\Cov}{Cov}
\DeclareMathOperator{\MSE}{MSE}
\DeclareMathOperator{\poly}{poly}
\DeclareMathOperator{\trace}{tr}
\DeclareMathOperator{\Var}{Var}

\usepackage[sc]{mathpazo}
\newtheorem{thm}{Theorem}[section]
\newtheorem{prop}[thm]{Proposition}

\theoremstyle{remark}
\newtheorem{remx}[thm]{Remark}
\newtheorem{examplex}[thm]{Example}

\theoremstyle{definition}
\newtheorem{defn}[thm]{Definition}

\newcommand{\qedblack}{\hfill \scalebox{0.7}{\ensuremath{\blacksquare}}}
\newenvironment{example}
  {\pushQED{\qedblack}\begin{examplex}}
  {\popQED\end{examplex}}

\newenvironment{rem}
  {\pushQED{\qedblack}\begin{remx}}
  {\popQED\end{remx}}
\renewcommand{\epsilon}{\varepsilon}

\makeatletter
\renewcommand{\@setemails}{%
  \mbox{\itshape E-mail address:\space}{\ttfamily\emails}%
}
\makeatother

\begin{document}

\title{Variance reduction methods in the estimation of Pauli sums}
\author{S\o ren Fuglede J\o rgensen$^1$}
\email{sfj@kvantify.dk}
\author{Rafael Emilio Barfknecht$^1$}
\author{Patrick Ettenhuber$^1$}
\author{Nikolaj Thomas Zinner$^{1,2}$}
\address{$^1$Kvantify, Rosenvængets Allé 25, DK-2100 Copenhagen, Denmark}
\address{$^2$Department of Physics and Astronomy, Aarhus University, DK-8000 Aarhus C, Denmark}

\date{\today}

\begin{abstract}
    Accurately estimating expectation values of quantum observables with as few measurements as possible is crucial to many quantum computing applications. We introduce a framework that covers many of existing measurement strategies and introduce heuristics that can be used to enhance randomized schemes, including those based on Pauli grouping with inverse probability weighting and variants of the classical shadow algorithm. We show how to maximize information gain from such schemes, while carefully optimizing the distribution of possible measurements, and show that simple grouping algorithms can get close to, and in some cases exceed, state-of-the-art accuracy for unbiased estimation of expectation values on a standard quantum chemistry benchmark. We show how these randomized methods may be compared to more recent measurement schemes, such as shadow grouping, derandomized shadow, and overlapped grouping measurement, we show how the same strategies can be used to augment these schemes, and we demonstrate that we can reduce measurement costs by up to a factor of two by allowing Clifford measurement circuits for otherwise Clifford-less methods.
\end{abstract}

\maketitle

\tableofcontents

\section{Introduction}

In the circuit model of quantum computation, the observables corresponding to physically relevant quantities are represented as multi-qubit operators. In general, the evaluation of expectation values of these observables in a certain quantum state, prepared using a given quantum circuit, is done by performing measurements of the state in a number of different measurement bases, obtained by decomposing the observables of interest into independently measurable parts.

The scaling of the number of measurements required to estimate an expectation value to a given precision can be unfavorable. For example, assume that the observable is a many-body quantum Hamiltonian, such as a molecular Hamiltonian equipped with a pre-defined basis set consisting of $n$ orbitals, then the Hamiltonian may be decomposed into a number of Pauli strings, typically growing like $O(n^4)$. Hence, without further optimization, the number of measurements required to estimate its expectation value to a given precision scales correspondingly. The amount of resources required to perform this estimation of course also depends on the required precision, and even for small molecules with modest basis sets, the resulting measurement counts can reach into the millions. What is more, in practical problems of interest, one may not know the relevant atomic configuration, or indeed the relevant state, a priori, and determining the state of interest itself involves repeated estimation of expectation values.

Our aim is to show how a sequence of heuristic algorithms can be combined to efficiently construct estimators which perform favorably compared to other available algorithms for the same problem. This work ties into a significant body of literature, with roots in shadow tomography \cite{Aaronson2020-tf}, and in work on graph colouring of Pauli non-commutation graphs: With the classical shadow algorithm \cite{Huang2020-ws, Elben2022-wz}, in which measurement procedures are randomized, one may efficiently measure large numbers of observables simultaneously. We will primarily be interested in measuring a polynomially-sized collection of known observables, and methods such as derandomized shadow \cite{Huang2021-ga} and locally-biased classical shadow \cite{Hadfield2022-vc, Zhang2023-lr} show that one may include the structure of the observables to improve on the concrete measurement procedures created by classical shadow. At the same time, algorithms based on optimally decomposing observables into smaller measurable parts, generally referred to as Pauli grouping algorithms, allow similar reductions in measurement requirements. \cite{Izmaylov2019-gd, Jena2019-hx, Yen2020-lk, Verteletskyi2020-qz, Zhao2020-xa, Gokhale2020-jc, Crawford2021-gz, Yen2023-zv, Burns2025-do} These include algorithms based on clique cover and colouring problem solvers which may be seen as defining static measurement schemes in which the choice of one measurement does not have any impact on which ones follow. Examples of adaptive schemes include \cite{Huang2021-ga,Hadfield2021-hx,Shlosberg2023-ok,Wu2023-sw,Liang2024-qd, Zhu2024-na, Gresch2025-kd, Gresch2025-zy,Li2025-vm}, which can also involve making favorable bias-variance tradeoffs during estimation. In all of this work one assumes that the observables of interest are represented as linear combinations of Pauli strings -- a native representation in the standard quantum circuit model -- but we note that one can take a wider perspective: For example, if the observable of interest is the Hamiltonian of a many-body fermionic system, this additional structure can be exploited, and one is led to perform fermionic tomography in which measurements are performed using matchgate circuits. \cite{Bonet-Monroig2020-qg, Jiang2020-ls, Zhao2021-gn, Wan2023-yf, Bian2025-au, Heyraud2025-zs} Our focus will be on measurement schemes in which measurements are performed using Clifford circuits and circuits consisting of only single qubit gates.

In the era of noisy quantum hardware, estimations of observables are typically used in hybrid quantum-classical algorithms that attempt to exploit the best of both sides \cite{bharti2022noisy}. Arguably, the most well-known case is 
 the variational quantum eigensolver (VQE, \cite{Peruzzo2014-uw}) in which a classical optimization algorithm searches the space of states guided by energy estimation \cite{tilly2022variational}. 
 More generally, this has spurred investigations into general variational quantum algorithms in which one typically has the issue that a number of classical parameters need to be optimized through extensive calls, and hence many measurements, to the quantum computer \cite{cerezo2021variational}. As in classical optimization tasks, a local strategy can often be advantageous, and in the case of molecular Hamiltonians a method such as ADAPT-VQE \cite{Grimsley2019-cp} or its variants such as FAST-VQE \cite{Majland2023-lz} has shown promise. In these local variational quantum algorithms, the optimization is guided by estimations of the derivatives of the energy in pre-determined directions. Here, classical shadow has been discussed as a means to avoid vanishing gradients \cite{sack2022avoiding,boyd2022training}. Another near-term application is the use of classical shadow techniques to measure wave function overlaps on a quantum computer that are subsequently used as input in auxiliary-field quantum Monte Carlo run on classical hardware to achieve precise results in quantum simulation \cite{huggins2022unbiasing,amsler2023quantum}. Moreover, the post-processing costs of this scheme has been reduced by using classical shadows based on random matchgate circuits \cite{Wan2023-yf}. More recent work has demonstrated further reductions in resources like qubit counts \cite{kiser2024classical}, and demonstrations in both superconducting \cite{huang2024evaluating} and ion traps have recently been reported \cite{zhao2025quantum}. In any of these hybrid schemes, the measurement overhead is always a key factor. 
 Consequently, there is great interest in trying to minimize the measurement requirements and to come up with algorithms that produce efficient measurement procedures. In the present work, our benchmarks will address the problem of estimations of sums of Pauli strings.

As an additional outlook, we note that classical shadow as a means to extract power from present noisy hardware has been explored for a handful of years now, starting with the demonstration of \cite{huang2022quantum}. In later studies, classical shadow has been used for tomography \cite{levy2024classical}, read-out error mitigation \cite{arrasmith2023development} and to deal with open quantum systems \cite{birke2026demonstrating}, as well as studies of how to use shadow techniques on platforms not based on the standard qubit paradigm, e.g. photonic quantum computers \cite{thomas2025shedding}. This line of study highlights the need to systematically deal with noise and errors in hardware \cite{chen2021robust}, and the question of how to verify results obtained using classical shadow \cite{karaiskos2025hard}. A large subfield of near-term quantum hardware developments is about mitigating errors, and classical shadow has found numerous uses here, both for present day hardware \cite{seif2023shadow,zhao2024group,wu2024error}, and for looking ahead towards the early fault-tolerant quantum devices that are currently under development \cite{jnane2024quantum,chan2025algorithmic}. In our discussions below, we will not consider the question of errors induced by noise but leave it for future work.

\subsection{Contributions of the paper}

In this paper, we treat the case of observables that are linear combinations of Pauli strings, and we carefully go through the known methods and the subproblems they aim to solve; we offer new perspectives on each of these, provide a comprehensive foundation for comparison, and show how they can be composed to provide new estimators. Along the way, we numerically benchmark all available options on the ground state energy estimation problem for a small library of standard molecules which has formed the de facto standard set of test instances for the problem.

Concretely, the estimators we consider are obtained as combinations of strategies for solving each of several subproblems; see Figure~\ref{fig:strategy-overview} for an overview.
\begin{figure}
    \centering
    \includegraphics[width=1\linewidth]{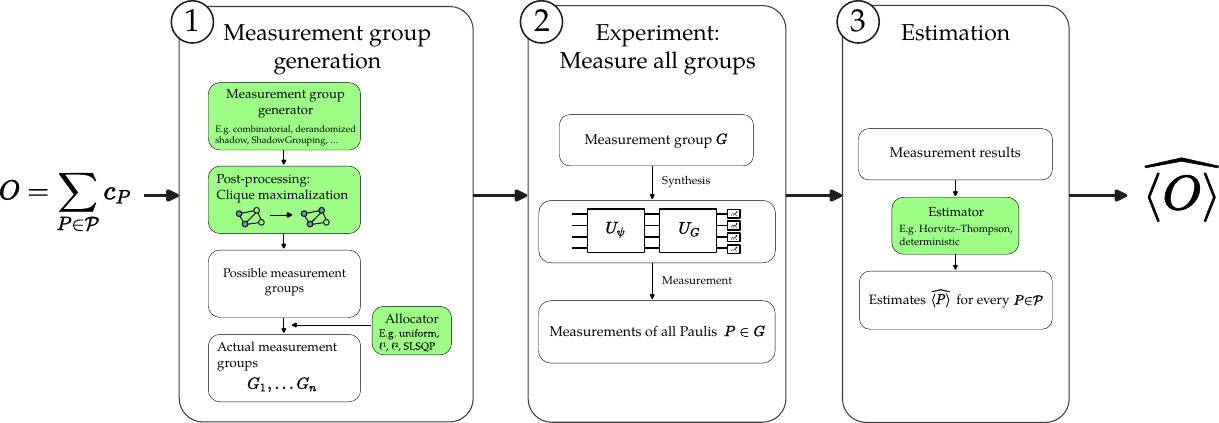}
    \caption{Overview of the strategies we consider: The observable $O$ of interest is used as input for a \emph{measurement group generator} to generate a collection of groups of Pauli strings. These groups may then be post-processed before an \emph{allocator} decides how many shots to associate to each group, or which probabilities to associate to each group in a randomized setting. Given a circuit $U_\psi$ preparing the state of interest, for each group $G$, one synthesizes a measurement circuit $U_G$, and by measuring the combined circuit, measurements for each Pauli string are collected and combined into an estimate $\estO$ of $\expecO$ in a way that depends on a choice of estimator type. The boxes highlighted in green indicate the subproblems for which we will consider several different strategies.}
    \label{fig:strategy-overview}
\end{figure}
\setcellgapes{3pt}  

\begin{itemize}
    \item \emph{Estimator type}: Earlier work on variance reduction methods for estimation of Pauli sums may be divided into that which produces randomized measurement schedules (using inverse probability weighting, or Horvitz--Thompson estimation) and that which relies on a schedule of $M > 0$ fixed measurements (in some cases by \emph{derandomizing} randomized schedules). We consider both variants (Section~\ref{sec:var-reduction}), and work to place them on the same footing, thereby providing a foundation for comparing estimators with different origins. This allows us to understand, for example, the benefit of using recent methods like ShadowGrouping \cite{Gresch2025-zy} compared to older ones like $\ell^1$ sampling (Section~\ref{sec:comparison-with-deterministic}).
    \item \emph{Measurement group generation}: Several heuristic algorithms for determining which measurements to perform exist. Throughout, we compare to known combinatorial methods, such as LDF grouping which is popular in practice, and several methods based on shadow estimation. We also introduce (Section~\ref{sec:stratified-sampling}) several new algorithms (ILP, G-SC, CG, LDVF, LVF) and find, for instance, that a structure-aware adaptation of LDF, which we call LVF (lowest variance first), offers a simple way to reduce resource requirements. (Section~\ref{sec:variance-aware-colouring}).
    \item \emph{Post-processing}: It is known that when used directly, measurement schemes based only on LDF grouping are wasteful, in the sense that they discard otherwise free information about estimates of summands of the observables being estimated. To alleviate this issue, and to put methods like LDF grouping on the same footing as other methods such as those based on shadow estimation (which output measurement bases rather than groups of Pauli strings), we consider a simple maximalization heuristic; we find that it can be used as a simple augmentation to greatly improve certain methods (Section~\ref{sec:unmax-vs-max-ldf}) and, in some benchmark instances, to obtain state-of-the-art results (Section~\ref{sec:comparison-results}). By allowing Clifford measurement circuits, we extend this maximalization procedure to what we call Cliffordization (Section~\ref{sec:cliffordization}), which can be used as a computationally cheap strategy to extend the otherwise Clifford-less methods like ShadowGrouping, effectively halving the number of measurements needed to achieve a given accuracy (Section~\ref{sec:comparison-results}).
    \item \emph{Allocation}: For Horvitz--Thompson estimators, and in inverse probability weighting methods in general, more useful measurements are included with higher probability, generally by relying on the structure of the observables of interest. Similarly, for deterministic estimators, one might want to assign different shot counts to different measurement bases. In both cases, we consider the problem of optimally allocating resources to the generated measurements. We find, for instance, that properly optimizing the probability can provide improvements over the $\ell^1$ sampling that has
    featured in earlier work (Section~\ref{sec:inv-prob-groups}). We also find that, once one conditions the randomized estimators on a particular design, this benefit seems to largely vanish (Section~\ref{sec:comparison-with-deterministic}). That is largely due to the nature of Horvitz--Thompson estimation, and more information about the structure of the quantum state is necessary to improve on allocation schemes (Section~\ref{sec:allocation-results}).
\end{itemize}

Note that in this framework, we may recover known methods as special cases; for example, the LDF grouping of \cite{Hadfield2021-hx} is the combination (Horvitz--Thompson, LDF, None, $\ell^1$), and the ShadowGrouping method of \cite{Gresch2025-zy} is (Deterministic, ShadowGrouping($M$), None, Uniform). This allows us to perform a structured comparison between all the different combinations, from which we obtain new evaluations of the usefulness of existing methods.

We find, for example, that out of the ones being considered, even though the more advanced measurement generators like ShadowGrouping are generally the most useful, the most accurate estimator for the ground state energy of \ce{BeH2} is one based on the LDF grouping (Section~\ref{sec:comparison-results}), a comparatively simple combinatorial method. On the other hand, we find that for measurement generators based on solving combinatorial problems, there is no clear relation between having a low number of distinct measurements, i.e. a low number of groups of Pauli strings, and the quality of the resulting estimator (Section~\ref{sec:comparison-of-grouping-methods}).

It also allows us to find new combinations that further reduce resource requirements: for example, we find that for the problem of estimating the ground state energy of \ce{NH3}, the lowest resource requirements are obtained from (Deterministic, ShadowGrouping($3\lvert \calP \rvert)$, Cliffordization, $\ell^2$) (Section~\ref{sec:comparison-results}, Table~\ref{tab:sota}).

\subsection{Structure of the paper}
The collection of all possible combinations of choices of strategies defines a very large collection of estimators. Rather than collecting all of them into one large comparison, we gradually build up the framework as outlined above, and provide results along the way, so that at each point we can determine which combinations of strategies can be useful, and which can be discarded.

In Section~\ref{sec:pauli-sums} we describe the general family of observables that we will handle. In Section~\ref{sec:benchmark} we introduce the benchmark molecules that will be used to guide the algorithm development. In Section~\ref{sec:var-reduction} we introduce the general statistical framework that helps define what we consider a ``good'' estimator. In Section~\ref{sec:stratified-sampling} we will see how the well-known framework of stratified sampling is a useful first step, and we study several algorithms for producing good strata. In Section~\ref{sec:importance-sampling} we turn to importance sampling, and specifically inverse probability weighting. In Section~\ref{sec:maximalization} we will see how in our case, partitioning into disjoint strata is generally informationally wasteful and we discuss possible remedies. In Section~\ref{sec:bias-variance} we show how introducing bias can be helpful in practice. Finally, in Section~\ref{sec:comparison-with-existing} we show how, without introducing bias and while having a static measurement procedure, the combination of these methods leads to estimators that are both unbiased and have low variance, and therefore to measurement procedures that require a lower number of measurements to reach a given target accuracy.

\subsection*{Acknowledgements}
The authors acknowledge fruitful discussion with Jaco van de Pol, Irfansha Shaik, Asbjørn Teilmann, Thomas Kj{\ae}rgaard, Kasper Poulsen, Thomas B{\ae}kkegaard, Sebastian Yde Madsen, and Janus Wesenberg.  
This work was funded in part by the European Innovation Council through Accelerator grant no. 190124924.
\section{Observables as Pauli sums}
\label{sec:pauli-sums}

Throughout, we will be dealing with $n$-qubit observables that are assumed to be given as a real linear combination of Pauli strings. We recall some common terminology. Concretely, let
\[
    I = \begin{pmatrix} 1 & 0 \\ 0 & 1 \end{pmatrix}, \quad
    X = \begin{pmatrix} 0 & 1 \\ 1 & 0 \end{pmatrix}, \quad
    Y = \begin{pmatrix} 0 & -i \\ i & 0 \end{pmatrix}, \quad
    Z = \begin{pmatrix} 1 & 0 \\ 0 & -1 \end{pmatrix}.
\]

\begin{defn}
A \emph{Pauli string} $P$ is any tensor product of $n$ of these four matrices, i.e. one of $4^n$ different operators acting on $(\bbC^2)^{\otimes n}$. As we will only be dealing with tensor products of those four matrices, we will use the common word notation $P = P_1 \cdots P_n = P_1 \otimes \cdots \otimes P_n$, $P_i \in \{I, X, Y, Z\}$.
\end{defn}
\begin{defn}
For a set $\calP$ of Pauli strings, with associated real numbers $c_P \in \bbR$, $P \in \calP$, we refer to the observable $O = \sum_{P \in \calP} c_P P$ as a \emph{Pauli sum}.
\end{defn}

We will be particularly interested in classes of Pauli sums that are sums of polynomially many terms, $\lvert \calP \rvert = \poly(n)$. For example, if the observable $O$ is the molecular Hamiltonian in an electronic structure problem, $\lvert \calP \rvert = O(n^4)$. This will play a role when considering the efficiency of algorithms for finding ``useful'' decompositions of $\calP$ later.
\begin{defn}
An $n$-qubit \emph{state} $\psi$ is a Hermitian, positive semidefinite operator on $(\bbC^2)^{\otimes n}$ satisfying $\trace(\psi) = 1$.
\end{defn}

We will be concerned with the problem of estimating the \emph{expectation value} of $O$ in the state $\psi$, defined by
\[
    \langle O \rangle = \trace(\psi O) = \sum_{P \in \calP} c_P \trace(\psi P).
\]

\subsection{Estimation through measurement}

We will focus only on estimation schemes that involve \emph{measurement}. A measurement of a Pauli operator $P$ in the state $\psi$ associates to $P$ a value in $\{-1, 1\}$, the eigenvalues of $P$, with probabilities given by the Born rule. It will be convenient to let $X_P$ denote the underlying $\{-1, 1\}$-valued random variable, so that in particular $\bbE [X_P] = \trace(\psi P)$. Note that the symbol $X_P$ is not related to the Pauli matrix $X$; we will use subscripts exclusively for the random variables.

It is possible to simultaneously measure a family of Pauli strings, as long as they pairwise commute. This produces a sample of several of the random variables; the variables are generally not independent, and their covariance is given by
\[
    \Cov(X_P, X_Q) = \trace(\psi PQ) - \trace(\psi P) \trace(\psi Q)
\]
for $P, Q \in \calP$.

The problem of measurement, or resource, optimization in the estimation of expectation values is the problem of producing a statistical estimator of $\langle O \rangle$ in which one is allowed a finite measurement budget of $M$ measurements, such that the estimator has a ``good'' accuracy; for instance, by requiring that the estimator produces an unbiased estimate and that it has low variance, or low squared error. Or, conversely, the problem of finding an estimator requiring a low number of measurements to achieve a given error with a given probability.

\subsection{Measurement circuits}
In the quantum circuit model, the simultaneous measurement of a family of commuting Pauli strings requires a circuit for preparing the state $\psi$, as well as a \emph{measurement circuit} which depends on the set of Pauli strings to be measured. For a general family of commuting Pauli strings, the measurement circuit can be taken to be a Clifford circuit; a circuit consisting of CNOT, $H$, and $S$ gates.

In practice, the cost of applying the circuit itself provides an important resource constraint. This means that if the resource cost of applying the measurement circuit is generally large compared to the cost of the state preparation, it can be desirable to try to restrict it to cases in which the measurement circuit itself requires only few resources. One important case is to restrict to sets of Pauli strings that qubit-wise commute:
\begin{defn}
Two Pauli strings $P = P_1 \cdots P_n$ and $Q = Q_1 \cdots Q_n$ are said to \emph{qubit-wise} commute if for every $i = 1, \dots, n$, $P_i$ commutes with $Q_i$.
\end{defn}
\begin{example}
    The Pauli strings $XX$ and $YY$ commute, but they do not qubit-wise commute since $X$ and $Y$ do not commute.
\end{example}
If the set of Pauli strings to be measured pair-wise qubit-wise commute, then the resulting measurement circuit can be taken to consist only of $H$ and $S$ gates, i.e. $O(n)$ 1-qubit gates. As CNOT gates represent a significant cost driver for current hardware, this simplification can be practically useful; the trade-off being that one is now generally restricted to measuring a smaller number of Pauli strings simultaneously, such that the total number of measurements required to reach a given error will a priori be larger. On the other hand, it is known that for any family of Clifford circuits, there are equivalent circuits with at most $O(n^2 / \log(n))$ CNOT gates \cite{Aaronson2004-ff} (and for the representations we will use, such Clifford circuits can be determined efficiently), meaning that if CNOT gates drive the cost, and if the number of CNOT gates required to prepare the state is, say, exponential, then an additional $O(n^2 / \log(n))$ gates may be considered relatively cheap.

In what follows, we will generally consider both the cases of qubit-wise commutation (QWC) and full commutation (FC), thus making clear what the trade-off will be.

\section{Benchmark: Electronic structure problem}
\label{sec:benchmark}

Throughout this work, we will motivate variance reduction methods and compare them to existing results in the literature by relying on a well-established benchmark based on molecular Hamiltonians of small molecules, and the problem of estimating their ground state energies, assuming that the ground state is prepared; that is, in these instances, $O = H$ is the Hamiltonian of the system, and $\psi$ denotes the ground state. We define all molecular configurations in the minimal basis set (STO-3G), with the exception of \ce{H2} which is also defined in the 6-31G basis. For each, we use the Jordan--Wigner mapping to encode the molecular Hamiltonian into a qubit Hamiltonian, and end up with the following systems: \ce{H2} (4 or 8  qubits), \ce{LiH} (12 qubits), \ce{BeH2} (14 qubits), \ce{H2O} (14 qubits) and \ce{NH3} (16 qubits). That allows us to increase the system size while still being able to obtain the exact ground state solution, i.e. an eigenvector of minimal eigenvalue, numerically. When comparing to existing results, we make sure, when possible, to use the exact same form of the Hamiltonian: The exact same coefficients and, to the extent it matters for a particular method, the exact same ordering of the Pauli strings; specifically, we rely on the Hamiltonians as provided in the supplementary material for \cite{Hadfield2022-vc}.

\begin{example}
    The smallest benchmark example, \ce{H2} in the STO-3G basis, is given, with three decimals of precision for each Pauli string, by
    \begin{align*}
        O = &-0.811\ IIII + 0.172\ ZIII -0.226\ IZII + 0.172\ IIZI -0.226\ IIIZ + 0.121\ ZZII \\
            &+ 0.0452\ YYXX + 0.0452\ YYYY + 0.0452\ XXXX + 0.0452\ XXYY \\
            &+ 0.169\ ZIZI + 0.166\ ZIIZ + 0.166\ IZZI + 0.175\ IZIZ + 0.121\ IIZZ
    \end{align*}
    From here, we may already make a few notes:
    \begin{itemize}
        \item  The coefficient for $IIII$ has the largest magnitude; this will play a role in variance reduction as $IIII$ also has the property that it commutes with all other Pauli strings, so while one can always include $IIII$ in any simultaneous measurement, doing so is irrelevant, as any measurement outcome of $IIII$ is $+1$. In other words, $\Var(X_{IIII}) = 0$ for any choice of state $\psi$.
        \item The distinction between qubit-wise and full commutation can matter in this example: The Pauli strings $\{YYXX, YYYY, XXXX, XXYY\}$ pairwise commute, but they do not pairwise qubit-wise commute.
    \end{itemize}
\end{example}

\section{Variance reduction}
\label{sec:var-reduction}
Recall that our goal is to define, for a fixed measurement budget $M$, an estimator of the expectation value $\langle O \rangle = \sum_{P \in \calP} c_P \bbE[X_P]$, and that we have at our disposal a choice of $M$ measurements of certain subsets of compatible Paulis. In other words, for each measurement round $r = 1, \dots, M$, and each Pauli string $P \in \calP$, we get an output $\{-1, 0, 1\}$, where we associate to $P$ the value $0$ if $P$ is not measured in measurement round $r$.

Then, an $M$-measurement \emph{estimator} of $\langle O \rangle$ is a real-valued function
\begin{align*}
    \widehat{\langle O \rangle} : (\{-1, 0, 1\}^\calP)^M \to \bbR.
\end{align*}
Its \emph{bias} is $\Bias(\estO) = \bbE[\estO] - \expecO$, and its \emph{variance} is $\Var(\estO) = \bbE[(\estO - \bbE[\estO])^2]$. The estimator is \emph{unbiased} if $\bbE[\estO] = \expecO$.

In what follows, we will often allow the choice of which measurement to perform in round $r$ to itself be a random variable $I$ whose sample space is the finite set $2^\calP$ of possible collections of Pauli strings; the resulting estimator then also depends on the sampled measurement, so its bias and variance will depend on $I$, and we will use the notation $\bbE$ to both mean expectation over the quantum state, or over this additional source of randomness. A useful fact is the law of total variance, which expresses the combined variance in terms of the purely quantum variance, i.e. the variance left over after conditioning on a particular collection of concrete measurements,
\[
    \Var(\widehat{\langle O\rangle}) = \bbE[\Var(\estO \mid I)] + \Var(\bbE[\estO \mid I]).
\]

\begin{example}
    Suppose that in each measurement round, we measure only a single Pauli string, and suppose that we follow a deterministic measurement schedule where each Pauli string $P \in \calP$ is measured a total of $M_P$ times, so that $\sum_{P \in \calP} M_P = M$, and assume that $M_P > 0$ for all $P$ (which in particular means that $M \geq \lvert \calP \rvert$). Let $n_P^\pm$ denote the number of times $\pm 1$ is measured for $P \in \calP$, and define an estimator by
    \[
        \estO^{\mathrm{det},M} = \sum_{P \in \calP} c_P \frac{n^+_P - n^-_P}{M_P}.
    \]
    Then, $\bbE[(n^+_P - n^-_P)/M_P] = \bbE[X_P] = \expecP$ and $\bbE[\estO^{\mathrm{det},M}] = \expecO$, so the estimator is unbiased. Its variance is
    \[
        \Var(\estO^{\mathrm{det},M}) = \sum_{P \in \calP} \frac{c_P^2}{M_P} \Var(X_P).
    \]
    From here, we read that if the goal is to minimize the $\Var(\estO)$, we should choose $M_P$ in a way that takes into account both the the coefficient $c_P$, but also $\Var(X_P)$. Here, $\Var(X_P)$ is unknown a priori, but we know, for example, that when $P$ is the identity, $\Var(X_P) = 0$, and so in this case allocating any measurements to the identity is wasteful.
\end{example}
\begin{example}
    \label{ex:horvitz-thompson}
    Suppose again that in each measurement round, we measure only a single Pauli string, but suppose now that we do so by drawing measurements according to some probability distribution; in other words, assume that we are given a probability distribution $\pi : \calP \to \bbR$ such that $\pi_P > 0$ for all $P \in \calP$, and define $I$ such that $\Pr(I = \{P\}) = \pi_P$. As before, let $n_P^\pm$ denote the actual measurements of $P$. Then, we can use inverse probability weighting and define the \emph{Horvitz--Thompson estimator}
    \[
        \estO_\pi = \frac{1}{M} \sum_{P \in \calP} c_P\frac{n^+_P - n^-_P}{\pi_P},
    \]
    which is now well-defined for all $M > 0$. As before, this is an unbiased estimator, and later we will see how to tune $\pi$ so as to minimize its variance.
\end{example}

In what follows, we will generally use the notation $\estO^{\mathrm{det},M}$ to refer to $M$-trial estimators for which the measurements to perform are fixed a priori, and $\estO_\pi$ to those for which the measurements are sampled from a probability distribution $\pi$. As indicated in the introduction, other choices govern which measurements are allowed in the first place; such choices will be implicit in the notation and spelled out in the text where relevant.
\section{Stratified sampling and partitioning algorithms}
\label{sec:stratified-sampling}

The previous examples illustrate estimators defined by measuring individual Pauli strings. We have already touched upon the fact that in practice, we can measure several Pauli strings simultaneously. In this section, we will discuss algorithms for partitioning $\calP$ into a small number of simultaneously measurable Pauli strings. In the literature, this is commonly referred to as ``Pauli grouping'' and is related to the graph colouring problem. Here, we spell out this correspondence and discuss a number of algorithms for solving it.

\subsection{Pauli commutation graphs}

We begin by noting how one can construct, for any such collection of Pauli strings, an undirected graph, which we call the \emph{Pauli commutation graph}, such that the group minimization problem becomes a minimum clique cover problem of that graph. To be able to handle both the cases of qubit-wise and full commutation, we will be introducing two such graphs, which we will refer to as the \emph{qubit-wise commutation graph} $G^{\QWC} = (V, E^{\QWC})$ and \emph{full commutation graph} $G^{\FC} = (V, E^{\FC})$ respectively.

In each case, the set of vertices is the set of all Pauli strings, $V = \calP$. The edges correspond to the type of commutativity, that is, $(P, Q) \in E^{\QWC}$ if $P$ and $Q$ qubit-wise commute, and $(P, Q) \in E^{\FC}$ if $P$ and $Q$ commute. We first note that adjacency can be checked efficiently.
\begin{prop}
    Let $P = P_1 \cdots P_n$ and $Q = Q_1 \cdots Q_n$, where $P_i, Q_i \in \{I, X, Y, Z\}$. Let
    \[
        \delta_i = \begin{cases}
        0, & \text{if $P_i = Q_i$ or $P_i = I$ or $Q_i = I$},\\
        1, & \text{otherwise,}
        \end{cases}
    \]
    and let $\delta = \sum_{i=1}^n \delta_i$. Then $(P, Q) \in E^{\QWC}$ if and only if $\delta = 0$, and $(P, Q) \in E^{\FC}$ if and only if $\delta$ is even.
\end{prop}

Note that any edge in the qubit-wise commutation graph is also an edge in the full commutation graph; the condition that two Pauli strings qubit-wise commute is stricter than only requiring them to commute.

The point of this construction is that a collection of Pauli strings can be simultaneously measured (resp. simultaneously measured with single-qubit measurement circuits) if and only if they form a \emph{clique} in $G^{\FC}$ (resp. $G^{\QWC})$.

To partition all Pauli strings into simultaneously measurable parts is therefore the same problem as finding a collection of cliques that cover all vertices of the graph. Finding a minimal such cover is called the minimal clique cover problem. We should note that a clique cover need not be a partition: cliques can overlap, and this will play an important role later.

\begin{rem}
    For an arbitrary undirected graph $G$, the smallest number of cliques necessary to cover the vertices of the graph is called its \emph{clique covering number} and is often denoted $\theta(G)$. Let $G^c$ denote the complement graph whose vertices are the vertices of $G$, and for which there is an edge between a pair of vertices if there is not an edge between those vertices in $G$. Then it is well-known that $\theta(G) = \chi(G^c)$, the \emph{chromatic number} of $G^c$, defined as the smallest number of colours by which one can colour the vertices of $G^c$ such that no two neighbouring vertices are coloured by the same colour. The translation is immediate: for a clique in a clique cover in $G$, colour all its vertices with the same colour to obtain a colouring by $\theta(G)$ colours in $G^c$. For this reason, in software packages for grouping Pauli strings one will often find that the problem is treated as a colouring problem in a non-commutation graph, rather than a clique cover problem in the commutation graph.
\end{rem}

Having the problem phrased as a well-known combinatorial optimization problem means that one can find a number of algorithms for solving the problem off-the-shelf, including approximation algorithms with known approximation ratios, i.e. proven guarantees on how far from optimality one will be.

\subsection{Algorithms for finding clique covers}

With the characterization of the grouping problem as a minimum clique cover problem in place, we can start pulling algorithms from the shelf. We will focus on four different algorithms, the first three of which assume that the set of maximal cliques is available to us. A \emph{maximal clique} in a graph is a clique such that there is no vertex that can be added to the clique to form a larger clique. Note that a priori, maximal cliques can have different cardinalities.

Throughout, let $G = (V, E)$ be a general graph, and let $\mathcal{C}$ denote the set of all maximal cliques in $G$.

\subsubsection{Integer programming}
\label{sec:ilp}
The minimum clique cover problem can be phrased as a simple integer program, which in fact solves the more general set cover problem. For every maximal clique $C \in \mathcal{C}$, we introduce a binary variable $x_C \in \{0, 1\}$, which takes the value $1$ if the clique $C$ is to be part of the cover. That is, at the end of the day, the goal will be to minimize $\sum_{C \in \mathcal{C}} x_C$. That the selection of cliques corresponds to a cover can be encoded by requiring that for every vertex $v \in V$, we must have $\sum_{\{C \in \mathcal{C} \mid v \in C\}} x_C \geq 1$; at least one of the cliques containing $v$ must be included. With this, any off-the-shelf integer programming solver can be used to produce minimal covers. In general, set cover is NP-hard, and from a practical point of view, as we will see below, only instances corresponding to very small graphs can be solved to optimality, but we note that solvers generally include options for halting when a given cover size is reached, or when a solution is known to have a given optimality gap.

\subsubsection{Greedy covering algorithms}
\label{sec:gsc}
Once $\mathcal{C}$ is known, a simple greedy algorithm for constructing covers is to start with an empty collection of cliques, then iteratively choose a clique $C \in \mathcal{C}$ which contains the largest number of vertices not already contained in a clique chosen earlier. This will produce a cover since in the worst case, we end up choosing all cliques.

It is well-known that the resulting cover will be at most a factor of $\sum_{k=1}^{\lvert V \rvert} \frac{1}{k}$ larger than a minimal cover.

\subsubsection{Column generation}
\label{sec:cg}
While the greedy algorithm above can be seen to be essentially the best polynomial time approximation algorithm for the general set cover problem, each iteration involves iterating over all maximal cliques, which itself can be prohibitively expensive in practice.

For example, if the observable is the molecular Hamiltonian such that $\lvert V \rvert$ will be of order $n^4$, then by the Moon--Moser theorem, the number of cliques can be up to $3^{n^4/3}$.

In general, when the number of variables in an integer program like the above grows exponentially with problem size, one can look towards column generation methods: methods that iteratively grow the set of variables at the expense of having to solve the problem several times.

Concretely, start with some clique cover $\mathcal{C}_0$, which could be a trivial solution like using all singletons, or a solution such as the one described in the section below which makes use of the structure of the Pauli graph and gives a better cover than the trivial one without also having to construct $\mathcal{C}$, and assume that we have constructed a set $\mathcal{C}_i$ of cliques. Then, we define $\mathcal{C}_{i+1} = \mathcal{C}_i \cup \{ C\}$, where $C$ is the clique defined as follows: First, solve the linear relaxation of the integer program from Section~\ref{sec:ilp} using $\mathcal{C}_i$ instead of $\mathcal{C}$, and allowing $x_{\tilde{C}} \in [0, 1]$ for all $\tilde{C} \in \mathcal{C}_i$. For the solution, calculate the marginals with respect to the inequality constraint; there is one constraint for every vertex $v \in V$, so each marginal is a real number $w_v > 0$. Use these weights to define a maximum weight clique problem, and let $C$ be a solution to this problem; that is, $C \in \mathcal{C}$ is a clique such that $\sum_{v \in C} w_v$ is maximal. This process continues until the new clique $C$ found this way already belongs to $\mathcal{C}_i$, at which point we solve the original integer linear program for the resulting collection of cliques.

The motivation for using the maximum weight clique problem to choose which cliques should enter the problem is that if the $x_C$ were continuous variables rather than binary, then this process would converge to an optimal solution of the corresponding linear program (and this is the setup usually considered when talking about column generation). In our case, we get no such guarantee, but we note that there are extensions that can guarantee optimality, most notably the branch and price algorithm in which one solves not the integer program in each step, but its linear relaxation, at the cost of getting solutions to the problem with non-integral values of the $x_C$, which are then used as branching variables. As such, the algorithm we consider may simply be thought of as branch and price with no branching.

There is a large body of literature on using branch and price to solve the maximum colouring problem, which as we have seen is equivalent to the clique cover problem, as well as variants like branch-price-and-cut, and in particular, on coming up with good branching rules; one classical set of rules that applies in our case are those of Ryan and Foster \cite{Foster1976-oa}. Based on the results we see below, we suspect that there is room for the application of this family of algorithms to Hamiltonians of interest here, but we make no such attempt in this study.

\subsubsection{Maximum weight clique solver}
\label{sec:gcg}
Note that the maximum weight clique problem is itself NP-hard, so that rather than having to solve one NP-hard problem, we now have to solve several. We find that depending on the size of the problem, this may not be practically feasible. 

For the examples below, we will consider two variants. One uses a variant of the branch and bound algorithm from \cite{bb-maxclique} as implemented in NetworkX \cite{networkx}. The other instead uses a greedy maximum weight solver, in which one iteratively picks the vertex with the highest value $f(v)$ for which there is an edge to all previously picked vertices. Here, the value of $f$ should somehow take into account the weight of the vertex, as well as its surrounding vertices. In our experiments, we define $f(v) = d(v) w(v)$, where $d(v)$ is the degree of $v$, and $w(v)$ is its weight. The quality of the resulting algorithm will of course depend on this choice of $f$, and we leave it as an open problem to investigate the usefulness of other heuristics.

\subsection{Greedy algorithms without constructing cliques}
\label{sec:gcn}
We round off our collection of algorithms by noting that it is also possible to greedily construct clique covers without having to first find all maximal cliques: first, choose some ordering of the vertices, $V = (v_1, v_2, \dots)$. We will now iterate over the vertices, and build up cliques as we go. Concretely, assume that the vertices $v_1, \dots, v_{i-1}$ have been associated to cliques $C_1, \dots, C_m$. To handle $v_i$, if there is a clique $C_j$ such that $C_j \cup \{v_i\}$ is a clique, then choose the first such clique $C_j$ and replace $C_j$ with $C_j \cup \{v_i\}$. If no such clique exists, introduce a new clique $C_{m+1} = \{v_i\}$.

As remarked above, one way to think about this approach is to realize that it corresponds to the natural greedy algorithm for the colouring problem on the complement graph $G^c$.

The resulting clique cover will depend on the choice of ordering of the vertices and the resulting cliques, and some heuristics could in principle be better than others; for example, one commonly encountered heuristic is to order the vertices such that those with lowest degree are handled first; note that this approach is commonly called ``\emph{largest} degree first (LDF)'' in the literature, when the focus is on solving the colouring problem in $G^c$. Numerous such heuristics for coming up with orderings have been studied \cite{Verteletskyi2020-qz}, and indeed this particular class of algorithms is generally used as a baseline for comparisons in studies on measurement optimization. Another heuristic we will consider below is ``recursive largest first (RLF)'' which, like LDF, starts a new clique $C_j$ by choosing the vertex with the largest degree in $G^c$, but subsequently and iteratively extends $C_j$ by choosing a vertex having the maximal number of neighbours in $G^c$ that are adjacent in $G^c$ to at least one element of $C_j$. \cite{Leighton1979-zw}

Note that it also avoids explicitly constructing the graph $G$, meaning that it is useful even in cases where the graph $G$ is too large to be stored in memory.

In all of the above algorithms, we have only used the commutativity structure of the Pauli strings. It is well established that insofar the aim is to estimate expectation values with low error, one should try to incorporate any available information, such as the coefficients of the Pauli strings in the observable, into the construction of the measurement scheme, and we will see in Section~\ref{sec:importance-sampling} how that may look for the algorithms at hand.

\subsection{Benchmark cases}

We now turn to the application of the above algorithms on our 6-molecule benchmark set, as defined in Section~\ref{sec:benchmark}.

In Table~\ref{tab:sizes}, we show the sizes of the resulting graphs and clique cover problems. Here, we write $\mathcal{C}^\mathrm{QWC}$ and $\mathcal{C}^\mathrm{FC}$ for the collections of maximal cliques in $G^{\QWC}$ and $G^{\FC}$ respectively.

\begin{table}
    \begin{tabular}{l|r|r||r|r||r|r}
        System & $n$ & $\lvert V \rvert$ & $\lvert E^\mathrm{QWC} \rvert$ & $\lvert\mathcal{C}^\mathrm{QWC}\rvert$  & $\lvert E^\mathrm{FC} \rvert$ & $\lvert\mathcal{C}^\mathrm{FC}\rvert$ \\
        \hline
        \ce{H2} (STO-3G) & 4 & 15 & 59 & 5 & 89 & 2 \\
        \ce{H2} (6-31G) & 8 & 185 & 4682 & 617 & 9804 & 2340 \\
        \ce{LiH} & 12 & 631 & 61077 & 6169 & 122493 & $> 10^6$ \\
        \ce{BeH2} & 14 & 666 & 71397 & 22933 & 143037 & $> 10^6$ \\
        \ce{H2O} & 14 & 1086 & 162629 & 91929 & 375099 & $> 10^6$ \\
        \ce{NH3} & 16 & 3057 & 1132908 & $> 10^6$ & 3085736 & $> 10^6$ \\
    \end{tabular}
    \bigskip
    \caption{The number of qubits and Pauli strings for various systems, and the numbers of edges and maximal cliques in the corresponding commutation graphs. The number of maximal cliques can be found by iteratively creating the cliques, and we only list results for those which have less than $10^6$ maximal cliques.}
    \label{tab:sizes}
\end{table}

We note immediately that $G^\mathrm{QWC}$ has relatively few edges -- its adjacency matrix is relatively sparse -- whereas $G^\mathrm{FC}$ has relatively many edges. In practice this means that when it comes to implementation, in the full commutation case it can be beneficial to work on the complement graph. Note that having more than $10^6$ maximal cliques means that the integer program in Section~\ref{sec:ilp} has more than $10^6$ variables.

\subsection{Experimental setup}
We now compare results and running times for five different algorithms. In the following,
\begin{itemize}
    \item ILP refers the result of solving the \textbf{i}nteger \textbf{l}inear \textbf{p}rogram in Section~\ref{sec:ilp} using the open source HiGHS solver \cite{Huangfu2018-sw}.
    \item G-SC refers to a Rust implementation of the \textbf{g}reedy algorithm for \textbf{s}et \textbf{c}over from Section~\ref{sec:gsc}.
    \item CG refers to the \textbf{c}olumn \textbf{g}eneration/branchless branch-and-price algorithm from Section~\ref{sec:cg}, using HiGHS to solve the linear programs, and using a Rust port of the NetworkX implementation of a maximum weight clique problem solver. For the starting solution $\mathcal{C}_0$, we use the trivial solution $\{\{v\} \mid v \in V\}$, and we record the result obtained after running the clique finding algorithm run for 30 minutes, followed by a 30 minute window for solving the resulting ILP.
    \item RLF refers to a Rust implementation of the \textbf{r}ecursive \textbf{l}argest \textbf{f}irst algorithm described in Section~\ref{sec:gcn}.
    \item LDF refers to a Rust implementation of the \textbf{g}reedy algorithm from Section~\ref{sec:gcn} with vertices ordered by \textbf{l}argest \textbf{d}egree \textbf{f}irst (largest degree in $G^c$, i.e. smallest degree in $G$).
\end{itemize}
All algorithms are generally single-threaded and are executed on an Intel(R) Core(TM) Ultra 7 155H.

\subsection{Benchmark results}
In Tables~\ref{tab:num-groups-qwc} and \ref{tab:num-groups-fc}, we report the numbers of cliques in the covers found by each of the above described algorithms, for the qubit-wise and full commutation graphs respectively. Note that in all cases, we remove the identity Pauli, which, as we have noted earlier, has no variance; it also commutes with all other Pauli strings, thus has no impact on any of the grouping algorithms described here.

The lowest clique count for each instance is highlighted in bold. Blanks denote problems that were too large to produce results within a few hours.

\begin{table}
    \begin{tabular}{l|r||r|r|r|r|r}
        System & $\lvert V \rvert$ & ILP & G-SC & CG & RLF & LDF \\
        \hline
        \ce{H2} (STO-3G) & 14 & \textbf{5}$^{\mathrm{o}}$ & \textbf{5} & \textbf{5} & \textbf{5} & \textbf{5} \\
        \ce{H2} (6-31G) & 184 & \textbf{37}$^{\mathrm{o}}$ & 45 & 38 & 38 & 46 \\
        \ce{LiH} & 630 & \textbf{131}$^{\mathrm{o}}$ & 142 & 132 & 133 & 136 \\
        \ce{BeH2} & 665 & \textbf{136}$^{\mathrm{o}}$ & 152 & 138 & 140 & 140 \\
        \ce{H2O} & 1085 & \textbf{210}$^{\mathrm{o}}$ & 250 & 213 & 219 & 224 \\
        \ce{NH3} & 3056 & & & & \textbf{609} & 618 \\
    \end{tabular}
    \bigskip
    \caption{The number of cliques in the covers of the \emph{qubit-wise commutation graph} resulting from applying the various algorithms. For example, for \ce{H2}, there are 14 Pauli strings, but 5 groups suffice to cover all of them. Here, $^\mathrm{o}$ denotes cases where the given result is optimal}
    \label{tab:num-groups-qwc}
\end{table}

\begin{table}
    \begin{tabular}{l|r||r|r|r|r|r}
        System & $\lvert V \rvert$ & ILP & G-SC & CG & RLF & LDF \\
        \hline
        \ce{H2} (STO-3G) & 14 & \textbf{2}$^{\mathrm{o}}$ & \textbf{2} & \textbf{2} & \textbf{2} & \textbf{2} \\
        \ce{H2} (6-31G) & 184 & \textbf{8}$^{\mathrm{o}}$ & 11 & \textbf{8} & \textbf{8} & 9 \\
        \ce{LiH} & 630 &   &   & \textbf{21}$^{\mathrm{to}}$ & 29 & 38 \\
        \ce{BeH2} & 665 &  &  & \textbf{20}$^{\mathrm{to}}$ & 31 & 33 \\
        \ce{H2O} & 1085 &  &  & \textbf{33}$^{\mathrm{to}}$ & 41 & 48 \\
        \ce{NH3} & 3056 & & & & 96 & \textbf{93} \\
    \end{tabular}
    \bigskip
    \caption{The number of cliques in the covers of the \emph{full commutation graph} resulting from applying the various algorithms. For molecules other than \ce{H2}, the number of maximal cliques is large enough that we only get results for the methods that do not rely on constructing cliques. Here, $^\mathrm{o}$ denotes cases where the given result is optimal, and $^\mathrm{to}$ denotes cases where results were obtained but for which the 30+30 minute time limit was reached.}
    \label{tab:num-groups-fc}
\end{table}

A number of remarks are in order:
\begin{enumerate}
    \item Only in the smallest cases will the full ILP be solvable. This is particularly true for the full commutation case, where the number of cliques is relatively larger. Indeed, simply finding all cliques in order to construct the resulting ILP is hard. In particular, this means that for only a small handful of cases do we have known optimality results.
    \item Similarly, since G-SC relies on being able to iterate over all cliques, this also only produces results in a few cases; when results are available, they are not in general better than those obtained from the other greedy method, LDF. Moreover, we suspect that when cliques can be generated within the time frame of relevance, one will generally also be able to simply solve the ILP, leaving little room for practical relevance for G-SC.
    \item For the qubit-wise commutation case, we find that CG is unable to find maximum weight cliques for problems of sizes where finding all cliques is hard, but in the full commutation case, there are situations where CG runs to completion and produces results that are better than all the included alternatives; in particular, having $21$ groups for \ce{LiH} could be an improvement over the $38$ groups provided by the commonly applied LDF method.
    \item LDF and RLF are the only algorithms for which results are available in all cases, as they do not require the generation of cliques. Even though RLF can be seen as a refinement over LDF, as expected from \cite{Verteletskyi2020-qz}, we note that it generally performs about the same, and in some cases quite a bit worse; RLF is also more computationally demanding than LDF, meaning that there is little reason to prefer it over LDF in practice.
    \item As a more general remark, note that the group counts for the full commutation graph are much lower than for the qubit-wise commutation graph. This matches expectation: In general, one expects that for an electronic structure Hamiltonian with $n^4$ Pauli summands, the optimal groupings would contain $O(n^4)$ and $O(n^3)$ groups in the qubit-wise and full commutation cases, respectively \cite[Table 1]{Huggins2021-al}.
\end{enumerate}

\subsection{The growth of Clifford circuits}
Recall that one reason for preferring qubit-wise commutation is that the resulting measurement circuits consist of single-qubit gates, as opposed to the Clifford circuits that one generally needs when dealing with full commutation. It is known \cite{Aaronson2004-ff} that $O(n^2 /\log n)$ gates suffice in general, but we may also ask exactly how many gates we need for the benchmark cases.

In Table~\ref{tab:clifford-circuit-growth}, we note the number of CNOT gates needed to synthesize the relevant Clifford circuits for the LDF groupings, averaged over all groupings for each of the benchmark molecules. The circuits are obtained from Stim \cite{gidney2021stim}, and we also record the result of simplifying those circuits with TKET \cite{Sivarajah_TKET_A_Retargetable_2020}. While we do not necessarily expect to be able to extrapolate from the $6$ molecule benchmark, we note that a least-squares fit of $\alpha \tfrac{n^2}{\log n}$ to the average CNOT gate counts post-simplification gives $\alpha \approx 0.449$; see Figure~\ref{fig:clifford-circuit-growth}.

\begin{table}
    \begin{tabular}{l||r|r}
        System & Average CNOT count (Stim) & Average CNOT count (Stim+TKET) \\
        \hline
        \ce{H2} (STO-3G) & 4.0 & 3.5 \\
        \ce{H2} (6-31G) & 31.4 & 15.1 \\
        \ce{LiH} & 45.5 & 24.9 \\
        \ce{BeH2} & 56.3 & 31.8 \\
        \ce{H2O} & 56.7 & 31.6 \\
        \ce{NH3} & 76.0 & 44.5 \\
    \end{tabular}
    \bigskip
    \caption{The average number of CNOT gates used in Clifford measurement circuits.}
    \label{tab:clifford-circuit-growth}
\end{table}

\begin{figure}
    \centering
    \includegraphics[width=0.75\linewidth]{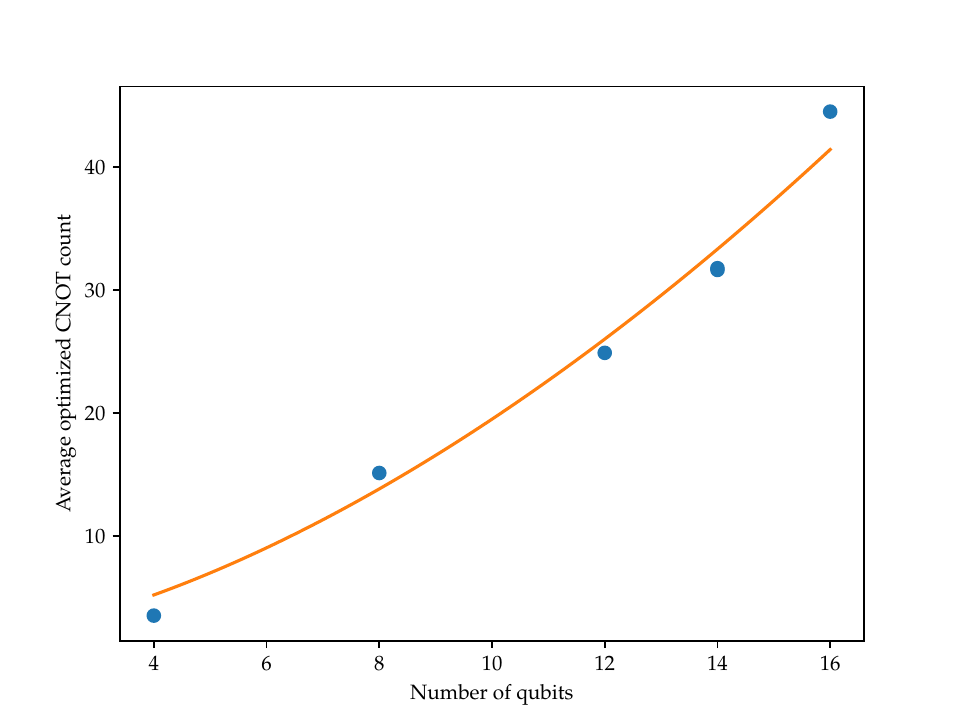}
    \caption{The average number of CNOT gates used in Clifford measurement circuits for the 6 benchmark molecules, together with the fit $n \mapsto 0.449 \tfrac{n^2}{\log n}$.}
    \label{fig:clifford-circuit-growth}
\end{figure}
\section{Importance sampling}
\label{sec:importance-sampling}
We have seen how partitioning methods allow us to gain more information in a single measurement by providing information about multiple Pauli strings at the same time, but we still need to decide how to use those possible measurements to estimate expectation values. In this section, we revisit the Horvitz--Thompson estimator from Example~\ref{ex:horvitz-thompson} in this context, and we show how a fixed collection of possible measurements, such as those provided by the partitioning algorithms in the previous sections, can be turned into such an estimator by letting the probabilities of performing a given measurement depend on the structure of the Hamiltonian. The derivations in this section are based on the Horvitz--Thompson estimator, but a similar story holds for the allocation of measurements in a deterministic estimator, which will be treated in Section~\ref{sec:derandomization-and-allocation}.

\subsection{Observable-aware sampling}

Some parts of the partition will be more useful than others, and we want to build any priors we have on usefulness into the estimator itself. Here, let us see how we may use the information about the observable itself to come up with a notion of usefulness.

\subsubsection{Inverse probability weighting for individual Pauli strings}
As a warm-up, let us momentarily forget that we were grouping together measurements of multiple Pauli strings and just consider the case where each measurement gives information about only one Pauli string (or, equivalently, say that we consider the trivial partition in which all parts have cardinality $1$). We may then associate to each Pauli string $P$ a weight $\pi_P > 0$ indicating how often it will be measured, and let us say that these weights are scaled relative to each other such that $\sum_P \pi_P = 1$. When post-processing the results of the measurement, to avoid biasing the resulting estimator, we need to take this reweighting into account; one way to do this is to scale the results from a single Pauli $P$ by $\frac{1}{\pi_P}$; we do what is often referred to as \emph{inverse probability weighting}. In other words, as our estimator of the expectation value for an $M$ sample experiment, we may take
\[
    \estO_\pi = \frac{1}{M} \sum_P  \frac{c_P}{\pi_P} \widehat{s_P},
\]
where $\widehat{s_P} = n_P^+ - n_P^-$ and $n_P^\pm$ is the number of times we observe $\pm 1$ for measurements of $P$; in other words, $\widehat{s_P}$ is an unbiased estimator of $M \trace(P\psi)$. The total estimator $\estO_\pi$ is also unbiased,
\[
    \mathbb{E}[\estO_\pi] = \sum_P c_P \mathbb{E}\left[\frac{\widehat{s_P}}{\pi_P M} \right] = \sum_P c_P \trace(P \psi) = \trace(\psi O),
\]
and we will want to choose $\pi_P$ such that we need as few measurements $M$ as possible, to get an estimate of a given quality. To this end, we focus on the one-shot variance of the estimator, i.e. the variance of the estimator obtained by setting $M = 1$. Note first that for $M = 1$, $X_P = n_P^+ - n_P^-$ is a random variable with values in $\{\pm 1\}$; let $\mu_P = \bbE[X_P] = \trace(P\psi)$ and $\sigma_P^2 = \Var(X_P)$. Let $I$ denote the random variable underlying the choice of Pauli, i.e. $I$ takes values in $\calP$, and $\Pr(I = P) = \pi_P$. For every $P \in \calP$,

\[
    \Var(\estO_\pi \mid I = P) = \Var\left(\frac{c_P}{\pi_P} \widehat{s_P} \right) = \frac{c_P^2}{\pi_P^2} \Var\left( \widehat{s_P} \right) = \frac{c_P^2}{\pi_P^2} \sigma_P^2,
\]
and so
\[
    \bbE[\Var(\estO_\pi \mid I)] = \sum_P \pi_P \Var(\estO_\pi \mid I = P) = \sum_P \frac{c_P^2}{\pi_P} \sigma_P^2.
\]

Similarly, for every $P \in \calP$, we have
\[
    \bbE[\estO_\pi \mid I = P] = \frac{c_P}{\pi_P} \mu_P,
\]
and, by letting $Y = \bbE[\estO_\pi \mid I]$, we have $\bbE[Y] = \sum_P \pi_P \bbE[\estO_\pi \mid I = P] = \sum_P c_P \mu_P$, and similarly $\bbE[Y^2] = \sum_P \pi_P \left(\frac{c_P\mu_P}{\pi_P}\right)^2 = \sum_P \frac{c_P^2}{\pi_P} \mu_P^2$, so it follows that
\begin{align*}
    \Var(\bbE[\estO_\pi \mid I]) &= \bbE[(\bbE[\estO_\pi \mid I])^2] - (\bbE[\bbE[\estO_\pi \mid I]])^2 = \sum_P \pi_P \left(\frac{c_P}{\pi_P}\mu_P\right)^2 - \left(\sum_P \pi_P \frac{c_P}{\pi_P} \mu_P\right)^2 \\
     &= \sum_P \frac{c_P^2}{\pi_P} \mu_P^2 - \left(\sum_P c_P \mu_P\right)^2.
\end{align*}
By the law of total variance, we have
\begin{align*}
    \Var(\estO_\pi) &= \bbE [\Var(\estO_\pi \mid I) ] + \Var(\bbE[\estO_\pi \mid I]) \\
     &= \sum_P \frac{c_P^2}{\pi_P} \sigma_P^2 + \sum_P \frac{c_P^2}{\pi_P} \mu_P^2 - \left(\sum_P c_P \mu_P\right)^2.
\end{align*}
Now $\sigma_P^2 + \mu_P^2 = \bbE[X_P^2] = 1$ since $X_P$ takes values in $\{-1, 1\}$, so that we end up with
\[
    \Var(\estO_\pi) = \sum_P \frac{c_P^2}{\pi_P} - \trace(\psi O)^2.
\]
Note that in particular, we may find the values of $\pi_P$ that minimize $\Var(\estO_\pi)$ despite not knowing the value of $\trace(\psi O)^2$ a priori: Define the Lagrangian $\mathcal{L} : \bbR_{>0}^{n_p} \times \bbR \to \bbR$ by
\[
    \mathcal{L}(\pi, \lambda) = \sum_P \frac{c_P^2}{\pi_P} + \lambda\left( \sum_P \pi_P - 1 \right).
\]
That $0 = \frac{\partial \mathcal{L}}{\partial \pi_P} = - c_P^2/\pi_P^2 + \lambda$ for all $P$ implies that $\pi_P = \sqrt{c_P^2/\lambda}$ for all $P$. In other words, $\pi_P$ is proportional to $\lvert c_P \rvert$ for all $P$, so the minimal posterior variance is obtained when
\[
    \pi_P = \frac{\lvert c_P \rvert}{\sum_P \lvert{c_P}\rvert}.
\]
Using these values of $\pi_P$, the resulting estimator is what is often referred to in the literature as the $\ell^1$ estimator (see e.g. \cite{Hadfield2022-vc}), named as such since the $\ell^1$ norm of the group -- in this case just a singleton -- is $\lvert c_P \rvert$. Let us immediately note that as far as singleton groups are concerned, since the $\ell^p$ norms on $\mathbb{R}$ agree for all $p > 0$, up until now there is no particular reason to favor $p = 1$.

\subsubsection{Inverse probability weighting for groups of Pauli strings}
\label{sec:inv-prob-groups}
Indeed, let us now proceed to perform the same analysis for partitions whose parts may consist of more than one Pauli and see that under the non-informative prior, weighting parts by their $\ell^2$ norms is optimal.

To this end, assume that we have partitioned our set of Pauli strings into a set of $n_G$ disjoint groups $\mathcal{G}$. Each group corresponds to a possible measurement, and as above, we want to associate to each group $G \in \mathcal{G}$ a probability $\pi_G > 0$ of performing the corresponding measurement. Such a measurement provides information about all the Pauli strings in the group, so just as before, we may use $\widehat{s_P} = n_P^+ - n_P^-$ as an estimator of the total of an individual Pauli string, and combine these into an estimator of the expectation value by
\[
    \estO_\pi = \frac{1}{M} \sum_{G \in \mathcal{G}} \frac{1}{\pi_G} \sum_{P \in G} c_P \widehat{s_P}.
\]
Note that the probability of sampling any particular Pauli $P$ is the probability $\pi_G$ of the group it belongs to. Therefore, as above, the estimator is unbiased,
\[
    \mathbb{E}[\estO_\pi] = \sum_{G \in \mathcal{G}} \sum_{P \in G} \mathbb{E}\left[\frac{c_P}{M\pi_G} \widehat{s_P} \right] = \sum_P c_P \trace(\psi P) = \trace(\psi O).
\]
In the above, we saw how the one-shot variance of the estimator could be expressed in terms of the variance of the individual Pauli strings; when grouping, we also need to take into account the potential covariance between Pauli strings within each group. Concretely, suppose again that $M = 1$, and write $X_P = n_P^+ - n_P^-$ be the $\{\pm 1\}$-valued variable for each Pauli string, and let $\mu_P = \bbE[X_P]$, and $\Sigma_{PQ} = \Cov(X_P, X_Q) = \trace(\psi P Q) - \trace(\psi P) \trace(\psi Q)$. As before, let $I$ denote the variable determining which group to measure, such that $\mathrm{Pr}(I = G) = \pi_G$. Then for each $G \in \mathcal G$, we have
\[
    \bbE[\estO_\pi \mid I = G] = \frac{1}{\pi_G} \sum_{P \in G} c_P \mu_P, \quad \Var(\estO_\pi \mid I = G) = \frac{1}{\pi_G^2} \sum_{P, Q \in G} c_P c_Q \Sigma_{PQ},
\]
and
\begin{align*}
    \Var(\bbE[\estO_\pi \mid I]) &= \sum_{G \in \mathcal G} \pi_G\left(\frac{1}{\pi_G}\sum_{P \in G} c_P \mu_P\right)^2 - \left(\sum_{G \in \mathcal G} \pi_G \frac{1}{\pi_G} \sum_{P \in G} c_P \mu_P \right)^2 \\
     &= \sum_{G \in \mathcal G} \frac{1}{\pi_G} \left(\sum_{P \in G} c_P\mu_P\right)^2 - \left(\sum_{G \in \mathcal{G}} \sum_{P \in G} c_P \mu_P\right)^2 \\
     &= \sum_{G \in \mathcal G} \frac{1}{\pi_G} \left(\sum_{P \in G} c_P\mu_P\right)^2 - \trace(\psi O)^2.
\end{align*}
By the law of total variance,
\begin{align*}
    \Var(\estO_\pi) &= \sum_{G \in \mathcal{G}} \frac{1}{\pi_G} \sum_{P, Q \in G} c_P c_Q \Sigma_{PQ} + \sum_{G \in \mathcal G} \frac{1}{\pi_G} \left(\sum_{P \in G} c_P\mu_P\right)^2 - \trace(\psi O)^2 \\
    &= \sum_{G \in \mathcal{G}} \frac{1}{\pi_G} \sum_{P, Q \in G} c_P c_Q (\Sigma_{PQ} +\mu_P\mu_Q) - \trace(\psi O)^2 \\
    &= \sum_{G \in \mathcal{G}} \frac{1}{\pi_G} \sum_{P, Q \in G} c_P c_Q \trace(\psi P Q) - \trace(\psi O)^2.
\end{align*}
Note that for any $G \in \mathcal G$, we have
\[
    \bbE\left[\left(\sum_{P \in G} c_P X_P\right)^2\right] = \sum_{P,Q \in G} c_Pc_Q(\Sigma_{PQ} + \mu_P\mu_Q),
\]
so we can also write
\[
    \Var(\estO_\pi) = \sum_{G \in \mathcal{G}} \frac{1}{\pi_G} \bbE\left[\left(\sum_{P \in G} c_P X_P\right)^2\right] - \trace(O\psi)^2.
\]
Unlike in the previous section, a variance-minimizing choice of $\{\pi_G\}_{G \in \mathcal G}$ will now depend on the state $\psi$, which is unknown a priori. To proceed we may assume the same non-informative prior on all individual Paulis, such that the expected covariance between distinct Pauli strings vanishes, and their variances are proportional. Under this assumption, up to an offset of $\trace(O\psi)^2$, and an overall factor given by the individual variances, the posterior variance becomes
\[
    \sum_{G \in \mathcal{G}} \frac{1}{\pi_G} \sum_{P,Q \in G} c_P c_Q \delta_{P,Q} = \sum_{G \in \mathcal{G}} \frac{1}{\pi_G} \sum_{P \in G} c_P^2.
\]
Define $\mathcal{L} : \bbR_{> 0}^{n_G} \times \bbR \to \bbR$ by
\[
    \mathcal{L}(\pi, \lambda) = \sum_{G\in\mathcal{G}} \frac{1}{\pi_G} \sum_{P \in G} c_P^2 + \lambda\left( \sum_{G \in \mathcal G} \pi_G - 1 \right).
\]
Now, for any group $G\in\mathcal{G}$, $0 = \frac{\partial L}{\partial \pi_G} = -\frac{1}{\pi_G^2} \sum_{P \in G} c_P^2 + \lambda$ implies that $\pi_G$ is proportional to $\sqrt{\sum_{P \in G} c_P^2}$, and that the posterior variance is minimized when
\begin{align}
    \label{eq:opt-non-overlapping}
    \pi_G = \frac{\sqrt{\sum_{P \in G} c_P^2}}{\sum_{G' \in \mathcal{G}} \sqrt{\sum_{P \in G'} c_P^2}}.
\end{align}

In other words, for disjoint groups, having probabilities be proportional to the $\ell^2$ norm of the groups is the optimal choice under the non-informative prior. Note that this is generally different from the distributions obtained based on the $\ell^1$ norm that has been used in some previous work. In Table~\ref{tab:compare-l1-and-l2} we give an indication of what this can mean in practice, by comparing the variances of the Horvitz--Thompson ground state energy estimators obtained by using the LDF groupings (see Section~\ref{sec:gcn}) for each of the 6 benchmark molecules in the qubit-wise case, for three different probability distributions; those obtained by taking $\pi_G$ to be proportional to $1$, to $\sum_{P \in G}\lvert c_P\rvert$, and to $\sqrt{\sum_{P \in G} c_P^2}$ respectively. We note that as expected, $\ell^2$ is superior in most cases. Note also that the results for $\ell^1$ exactly matches those reported in \cite{Hadfield2021-hx}.

\begin{table}
    \begin{tabular}{l||r|r|r}
        System & Uniform & $\ell^1$ & $\ell^2$ \\
        \hline
        \ce{H2} (STO-3G) & 4.17 & \textbf{0.402} & 0.424 \\
        \ce{H2} (6-31G) & 197 & 22.3 & \textbf{21.0} \\
        \ce{LiH} & 576 & 54.2 & \textbf{46.6} \\
        \ce{BeH2} & 2470 & 135.0 & \textbf{117} \\
        \ce{H2O} & 66000 & 1040 & \textbf{921} \\
        \ce{NH3} & 96700 & 891 & \textbf{732} \\
    \end{tabular}
    \bigskip
    \caption{The variances of the Horvitz--Thompson estimators of LDF grouped Hamiltonians using different probability distributions for the groups. All values are in units of Ha$^2$.}
    \label{tab:compare-l1-and-l2}
\end{table}

\subsection{Variance-aware colouring algorithms}
\label{sec:variance-aware-colouring}
We have seen that LDF colouring provides a fast and simple heuristic for producing clique covers of the Pauli commutation graphs. It does so, however, in a way that is oblivious to the fact that we will want to minimize variances at the end of the day. Therefore we may ask if there are simple greedy colouring algorithms that keep the end goal in mind.

To this end, we will want at each iteration to be able to decide if adding a Pauli string to a clique will reduce the resulting variance of the associated Horvitz--Thompson estimator. We can achieve this by defining, for a collection $\{G_1, \dots, G_m\}$ of cliques a heuristic $h(\{G_1, \dots, G_m\}) = \Var(\estO_\pi)$, where $\pi$ is given by \eqref{eq:opt-non-overlapping}.

Based on this, we will consider two variations on LDF; in both cases, assume that the Pauli strings have been ordered $P_1, \dots, P_k$ by degree. We iteratively add a Pauli string to a clique, so at each iteration, we will have a set $\{G_1, \dots, G_m\}$ of cliques.
\begin{itemize}
    \item In ``Lowest Degree, then Variance, First'' (LDVF), assume that the Pauli strings $P_1, \dots, P_{i-1}$ have already been grouped into cliques. To include $P_i$, evaluate for every $j = 1, \dots, m$, the heuristic $h(\{G_1, \dots, G_{j-1}, G_j \cup \{P_i\}, G_{j+1}, \dots, G_m\})$ as well as $h(\{G_1, \dots, G_m, \{P_i\}\})$, i.e. the estimated variance obtained by adding $P_i$ to one of the existing cliques, or to its own new clique. Use the set of cliques leading to the lowest estimated variance. In the case of a draw, use the grouping that was evaluated first.
    \item In ``Lowest Variance First'' (LVF), assume that some subset $\calP' \subseteq \calP$ has been grouped. Then for every $i = 1, \dots, k$, if $P_i$ is not already in $\calP'$, evaluate for every $j=1, \dots, m$ the estimated variance by adding $P_i$ to one of the cliques, or to its own new clique, as above. Keep the set of cliques leading to the lowest estimated variance, once again with draws being broken by the order in which the variances have been evaluated.
\end{itemize}
In other words, LDVF handles the vertices in the same order as LDF, whereas LVF spends some additional effort by trying all remaining Pauli strings in each iteration.

In Table~\ref{tab:compare-ldf-ldvf-lvf}, we provide a comparison of these three methods. Concretely, the table lists the variances of the Horvitz--Thompson estimators obtained by using the groupings from each of the three methods, with probabilities given by $\ell^2$ scaling \eqref{eq:opt-non-overlapping} for both the QWC and FC graphs. We note that LDVF or LVF outperform LDF in almost all cases.
\begin{table}[h!]
    \centering
    \begin{tabular}{l||r|r|r||r|r|r}
        \multicolumn{1}{c||}{} & 
        \multicolumn{3}{c||}{QWC} & 
        \multicolumn{3}{c}{FC} \\
        System & LDF & LDVF & LVF & LDF & LDVF & LVF \\
        \hline
        \ce{H2} (STO-3G) & \textbf{0.424} & \textbf{0.424} & \textbf{0.424} & \textbf{0.352} & \textbf{0.352} & 0.364 \\
        \ce{H2} (6-31G) & 21.0 & \textbf{15.8} & 26.4 & 14.8 & 9.12 & \textbf{5.89} \\
        \ce{LiH} & 46.6 & \textbf{20.2} & 34.5 & \textbf{18.4} & 18.5 & 23.1 \\
        \ce{BeH2} & 117 & \textbf{46.5} & 69.1 & 102 & 75.9 & \textbf{74.3} \\
        \ce{H2O} & 921 & \textbf{450} & 590 & 1070 & 611 & \textbf{430} \\
        \ce{NH3} & 732 & 570 & \textbf{503} & 625 & \textbf{315} & 470 \\
    \end{tabular}
    \bigskip
    \caption{Comparison of LDF, LDVF, and LVF in terms of the variances of the Horvitz--Thompson estimators of the grouped Hamiltonians. All values are in units of Ha$^2$.}
    \label{tab:compare-ldf-ldvf-lvf}
\end{table}

\subsection{Classical shadow and its variants}

At this point, we could return to the clique covers of Section~\ref{sec:stratified-sampling} and evaluate the variances of the estimators obtained by using them in randomized measurement schemes. However, above, we assumed that the groups of Pauli strings formed a partition, a cover by disjoint cliques, and whereas algorithms like those based on colouring the complement graph do produce partitions, some of the clique cover algorithms, e.g. the one based on integer linear programming, will produce covers by overlapping cliques.

One can always turn a clique cover of overlapping cliques into one of the same cardinality that uses only disjoint cliques by selecting, for each Pauli belong to more than one clique, one of those cliques at random. But recall that our goal is ultimately to use the cliques to estimate expectation values, so by trimming the cliques we are effectively reducing the amount of information we gain about each Pauli, which is counterproductive. For example, if $\calP = \{XY, XI, XZ\}$, a possible partition is $\{\{XY, XI\}, \{XZ\}\}$, but in measuring $XZ$ it is possible to also gain information about $XI$.

In Section~\ref{sec:maximalization} below, we turn this around and instead look at ways to turn clique covers into clique covers by maximal cliques.

It may be useful to contrast the clique cover approach with some of the other approaches to measurement optimization. In the clique cover algorithms, we are first and foremost concerned with grouping Pauli strings together, and then only later will we consider what the measurement circuit corresponding to a given group should be, and how to sample those measurement circuits. This is in contract, for instance, to the approach in the classical shadow \cite{Huang2020-ws} algorithm in which one samples not groups of Pauli strings, but measurement circuits themselves, and once a measurement circuit is chosen, one can look at which Pauli strings it provides information about; and by doing so, one ensures automatically that the resulting clique of Pauli strings is maximal. Now in classical shadow, the aim is to simultaneously learn about very many observables, but once a fixed observable, or a fixed set of observables, is known, one can look at algorithms for producing measurement circuits that play well with those observables. This is the approach taken in e.g. \cite{Hadfield2021-hx,Hadfield2022-vc,Shlosberg2023-ok,Gresch2025-zy}. 
\section{Maximalization: From partitions to overlapping groups}
\label{sec:maximalization}

So it is clear that focusing on algorithms that natively produce maximal cliques should be a more natural fit for the purposes of estimating expectation values. Nevertheless, one may attempt to convert non-maximal cliques into maximal ones to see if that leads to equally useful methods in practice. In this section, we do just that, but to begin with, we investigate what impact this may have on the optimal choice of measurement circuit probability distribution: previously, if a particularly important Pauli string belongs to a given group, it would make sense to measure that group with a correspondingly higher frequency, but once the Pauli string belongs to more than one group, its impact on the probability would have to be spread out accordingly. As above, we perform the derivation for Horvitz--Thompson estimators first, and return to deterministic estimators later.

\subsection{Estimation for overlapping groups}
\label{sec:estimation-overlapping-groups}
So let us now revisit the choice of estimator made in Section~\ref{sec:inv-prob-groups} in the case where groups overlap. Let $\mathcal{G}$ denote a set of groups whose union is all Pauli strings in the observable, and let $\pi$ be a probability distribution on $\mathcal{G}$. For each Pauli string $P$, let $\mathcal{G}_P \subseteq \mathcal{G}$ denote the groups to which $P$ belongs, and define $\pi_P = \sum_{G \in \mathcal{G}_P} \pi_G$. That is, if in a given iteration of our measurement procedure we sample a group $G \in \mathcal{G}$ with probability $\pi_G$, then $\pi_P$ is the probability of learning about $P$ in that iteration. Moreover, if this procedure has been carried out $M$ times, let $n_{P,G}^\pm$ denote the number of times $\pm 1$ has been observed for measurements of $P$ in measurements corresponding to a group $G \in \mathcal{G}_P$, and let $n_G = n_{P,G}^+ + n_{P,G}^-$ denote the number of times group $G$ has been sampled. As before, introduce
\[
    \widehat{s_{P,G}} = n_{P,G}^+ - n_{P,G}^-
\]
where now $n_{P,G}^\pm$ denote the number of times we observe $\pm 1$ for $P$, based only on observations corresponding to $G$. With this, define an estimator $\widehat{s_P}$ of the total of $P$ by
\[
    \widehat{s_P} = \frac{1}{\pi_P} \sum_{G \in \mathcal{G}_P} \widehat{s_{P,G}}
\]
and an estimator of the expectation value of $O$ by
\[
    \estO_\pi = \frac{1}{M} \sum_{P} c_P \widehat{s_P}.
\]
Since $\bbE[\widehat{s_{P,G}}] = M\pi_G \trace(\psi P)$, we have
\[
    \bbE[\widehat{s_P}] = \frac{1}{\pi_P} \sum_{G \in \mathcal{G}_P} M\pi_G \trace(\psi P) = \frac{M\trace(\psi P)}{\pi_P} \sum_{G \in \mathcal{G}_P} \pi_G = M\trace(\psi P),
\]
and therefore finally
\[
    \bbE[\estO] = \frac{1}{M}\sum_P c_P M \trace(\psi P ) = \sum_P c_P \trace(\psi P) = \trace(\psi O ),
\]
so $\estO_\pi$ is unbiased.

Even though measurements from different groups can be used to infer information about the same Pauli string, two measurements of the same Pauli coming from different trials are independent, so that at the end of the day, the variance of $\estO_\pi$ does not depend on the covariance of the same Pauli string between different groups. Concretely, assume again that $M = 1$, and let us determine the one-shot variance. As previously, let $X_P = n_{P}^+ - n_P^-$ be the variable counting measurement results for a Pauli string $P$, let $\mu_P = \bbE[X_P]$ and $\Sigma_{PQ} = \Cov(X_P, X_Q)$, and note that even though $P$ may belong to several groups, when $M = 1$ at most one such group is measured. If one samples a given group $G \in \mathcal{G}$, the only non-trivial contributions to $\hat\mu$ come from the Pauli strings $P$ such that $\mathcal G_P$ is non-empty; indeed, the total contribution is
\[
    \sum_{\{P \mid G \in \mathcal{G}_P\}} \frac{c_P}{\pi_P} (n^+_{P} - n_{P}^-).
\]
Note that $\{P \mid G \in \mathcal G_P\} = G$. Let $I$ denote the group-valued random variable with $\Pr(I = G) = \pi_G$. Then
\[
    \bbE[\estO_\pi \mid I = G] = \sum_{P \in G} \frac{c_P\mu_P}{\pi_P}.
\]
Similarly,
\[
    \Var(\estO_\pi \mid I = G) = \sum_{P,Q \in G} \frac{c_Pc_Q}{\pi_P\pi_Q} \Sigma_{PQ},
\]
and
\[
    \Var(\bbE[\estO_\pi \mid I]) = \sum_{G \in \mathcal G} \pi_G \left(\sum_{P \in G}\frac{c_P\mu_P}{\pi_P}\right)^2 - \left(\sum_{G \in \mathcal G} \pi_G \sum_{P \in G} \frac{c_P\mu_P}{\pi_P}\right)^2.
\]
As before, by reordering the terms, the final term reduces to $\trace(\psi O)^2$ since
\[
    \sum_{G \in \mathcal G} \pi_G \sum_{P \in G} \frac{c_P\mu_P}{\pi_P} = \sum_P \frac{c_P\mu_P}{\pi_P} \sum_{G \in \mathcal G_P} \pi_G = \sum_P c_P \mu_P = \trace(\psi O).
\]
On the other hand,
\[
    \bbE[\Var(\estO_\pi \mid I)] = \sum_{G \in \mathcal G} \pi_G \sum_{P,Q \in G} \frac{c_Pc_Q}{\pi_P\pi_Q} \Sigma_{PQ},
\]
so the total variance is
\begin{align*}
    \Var(\estO_\pi) &= \sum_{G \in \mathcal{G}} \pi_G \sum_{P, Q \in G} \frac{c_Pc_Q}{\pi_P\pi_Q} (\Sigma_{PQ} + \mu_P \mu_Q) - (\trace(\psi O))^2 \\ 
     &=\sum_{G \in \mathcal{G}} \pi_G \sum_{P, Q \in G} \frac{c_Pc_Q}{\pi_P\pi_Q} \trace(\psi PQ) - (\trace(\psi O))^2.
\end{align*}
And as before, if we introduce for each group $G$ the sub-Hamiltonian $H_G = \sum_{P \in G} \frac{c_P}{\pi_P} P$, this can also be written as
\[
    \Var(\estO_\pi) = \sum_{G \in \mathcal G} \pi_G \bbE[H_G^2] - \trace(\psi O)^2.
\]
As before, we can not in general determine variance-minimizing $\{\pi_G\}_{G \in \mathcal G}$ without knowing the state $\psi$, but if we assume a non-informative prior, we may produce an approximate solution by instead minimizing
\begin{align}
    \label{eq:opt-problem-overlapping-groups}
    \sum_{G \in \calG} \pi_G \sum_{P \in G}\frac{c_P^2}{\pi_P^2}.
\end{align}

This time the Lagrangian $\mathcal L: \bbR^{n_G}_{> 0} \times \bbR \to \bbR$ becomes
\[
    \mathcal{L}(\pi, \lambda) = \sum_{G \in \mathcal G} \pi_G \sum_{P \in G} \frac{c_P^2}{\pi_P^2} + \lambda\left(\sum_{G \in \mathcal{G}} \pi_G - 1 \right).
\]
By reordering the terms,
\[
    \sum_{G \in \mathcal{G}} \pi_G \sum_{P \in G} \frac{c_P^2}{\pi_P^2} = \sum_{P} \frac{c_P^2}{\pi_P^2} \sum_{G \in \mathcal G_P} \pi_G = \sum_P \frac{c_P^2}{\pi_P}.
\]
Now $\tfrac{\partial \pi_P}{\partial \pi_G}$ is $1$ if $G \in \mathcal G_P$, i.e. the $G$ for which $P \in G$, and vanishes otherwise. In particular, for a $G \in \mathcal G$, by the chain rule,
\[
    \frac{\partial \mathcal L}{\partial \pi_G} (\pi, \lambda) = -\sum_{P \in G} \frac{c_P^2}{\pi_P^2} + \lambda,
\]
so at stationary points we must have $\lambda = \sum_{P \in G} \frac{c_P^2}{\pi_P^2}$ for all $G \in \mathcal G$. By applying this for every $G \in \mathcal G$, we see that
\[
    \lambda = \lambda \sum_{G \in \mathcal G} \pi_G = \sum_{G \in \mathcal G} \lambda \pi_G = \sum_{G \in \mathcal G} \pi_G \sum_{P \in G} \frac{c_P^2}{\pi_P^2},
\]
which can be used as a certificate of optimality. However, unlike what we have seen in the previous special cases, no closed form solution exists.

Instead, in our numerical experiments we rely on off-the-shelf constrained optimization solvers, and for the results reported below, we use the SciPy SLSQP solver \cite{SciPy2020-NMeth, Lawson1995-gl, Kraft1988}, with its default convergence criteria. We find that for the benchmark cases, it is important to provide a non-trivial initial guess $\pi^0$, which we take to be
\[
    \pi_G^0 = \frac{\sqrt{\sum_{P \in G} c_P^2}}{\sum_{G' \in \mathcal{G}} \sqrt{\sum_{P \in G'} c_P^2}},
\]
i.e. what would be the optimal solution if the groups did not overlap.

\subsection{LDF maximalization}

Recall from Section~\ref{sec:gcn} that one may obtain clique covers by finding colourings of the complement of the Pauli commutation graph in a largest degree first fashion.

These covers are partitions but may be extended to covers by maximal cliques in a similar fashion: once a cover $C_1, \dots, C_m$ has been found, order the vertices, i.e. the Pauli strings, as $P_1, \dots, P_k$ such that the for $i \leq j$, the degree of $P_i$ in the Pauli commutation graph is lower than the degree of $P_j$. Now, for $i = 1, \dots, k$ and $j = 1, \dots, m$, add $P_i$ to $C_j$ if $C_j \cup \{P_i\}$ is a clique.

Note that this process is well-defined for any clique cover; not just one that is obtained from largest degree first colouring, and we refer to this heuristic for converting clique covers to covers by maximal cliques as \emph{LDF maximalization}.

\subsubsection{Cliffordization}
\label{sec:cliffordization}
Notably, when performing this maximalization, the Pauli commutation graph could be either $G^{\QWC}$ or $G^{\FC}$. This also means that one can apply the procedure to a clique cover of $G^{\QWC}$ and get a cover by maximal cliques in $G^{\FC}$.

One place where this may be useful is if one works in a setting where applying general Clifford measurement circuits is unproblematic from a resource perspective, but where the procedure for producing cliques is one which natively works over $G^{\QWC}$; that is, any procedure that produces Pauli basis measurements. Such procedures include many of the classical shadow based methods \cite{Huang2020-ws, Huang2021-ga, Hadfield2021-hx, Hadfield2022-vc, Gresch2025-zy}, in which the output of the measurement generation scheme is simply a collection of Pauli measurement operations, which may be thought of as cliques in $G^{\QWC}$, by identifying any given measurement circuit with the Pauli strings for which the measurement provides information. As passing from $G^{\QWC}$ or $G^{\FC}$ is exactly the process of having to apply general Clifford measurement circuits, we refer to this approach to augmenting an existing measurement schedule as \emph{Cliffordization}, and in Section~\ref{sec:comparison-with-existing}, we will see, for example, the improvements obtained by applying this Cliffordization on top of ShadowGrouping \cite{Gresch2025-zy}.

\subsection{Benchmark results}
We are now in a position where we can test whether maximalization offers a benefit on the benchmark cases, and where we may now finally compare the different grouping methods from Section~\ref{sec:stratified-sampling} in terms of the variances of the resulting estimators.

\subsubsection{Unmaximalized vs. maximalized LDF}
\label{sec:unmax-vs-max-ldf}
In Table~\ref{tab:compare-ldf-maximalization} we show the variance of the estimators obtained using the LDF grouping (Section~\ref{sec:gcn}) which gives a partition of $\calP$, using the $\ell^2$ probability distribution \eqref{eq:opt-non-overlapping}. We then use LDF maximalization on top of this grouping to obtain a cover by maximal cliques. To turn this into a complete estimator, we take as our probability distribution on one hand the $\ell^2$ distribution given by \eqref{eq:opt-non-overlapping}, and on the other the one optimized using SLSQP as described in Section~\ref{sec:estimation-overlapping-groups}. While not necessarily optimal in the case of overlapping groups, the $\ell^2$ distribution can be calculated efficiently and functions as a baseline for further optimization. We find that for the larger molecules, even this naive baseline provides an improvement over the unmaximalized case, and that as expected, using the off-the-shelf SLSQP solver leads to even greater improvements. Note also that in several cases, the variances for the estimators from QWC grouping are smaller than those for FC grouping, even though we generally expect the added flexibility from being able to use the larger FC groups (and the correspondingly more resource intensive measurement circuits) would make it possible to do with fewer measurements. 

\begin{table}[h!]
    \centering
    \begin{tabular}{l||r|r|r||r|r|r}
        \multicolumn{1}{c||}{} & 
        \multicolumn{3}{c||}{QWC} & 
        \multicolumn{3}{c}{FC} \\
        System & $\ell^2$ & Maxim. + $\ell^2$ & Maxim. + SLSQP &$\ell^2$ & Maxim. + $\ell^2$ & Maxim. + SLSQP \\
        \hline
        \ce{H2} (STO-3G) & \textbf{0.424} & \textbf{0.424} & \textbf{0.424} & 0.352 & 0.493 & \textbf{0.288} \\
        \ce{H2} (6-31G) & 21.0 & 10.4 & \textbf{5.48} & 14.8 & 16.8 & \textbf{13.9} \\
        \ce{LiH} & 46.6 & 17.8 & \textbf{7.2} & 18.4 & 20.9 & \textbf{9.67} \\
        \ce{BeH2} & 117 & 42.0 & \textbf{23.2} & 102.0 & 30.8 & \textbf{15.1} \\
        \ce{H2O} & 921 & 607 & \textbf{131} & 1070 & 576 & \textbf{291} \\
        \ce{NH3} & 732 & 551 & \textbf{172} & 625 & 219 & \textbf{145} \\
    \end{tabular}
    \bigskip
    \caption{Comparison of estimator variances for unmaximalized LDF grouping using the $\ell^2$ probability distribution, and maximalized LDF grouping using either the $\ell^2$ probability distribution or improving on it using SLSQP. All entries are in units of [Ha$^2$].}
    \label{tab:compare-ldf-maximalization}
\end{table}

\subsubsection{Comparison of grouping methods}
\label{sec:comparison-of-grouping-methods}
We are now finally in a position where we can determine if the clique cover minimalization problem serves as a good proxy for variance reduction. Recall that in Tables~\ref{tab:num-groups-qwc} and \ref{tab:num-groups-fc}, we compared the clique covers obtained from a variety of grouping methods. In Tables~\ref{tab:grouping-algorithm-variance-qwc} and \ref{tab:grouping-algorithm-variance-fc}, we give the corresponding variances. Here, in every case, we use maximalization when the clique cover is not already a cover by maximal cliques, and we use SLSQP to optimize the probability distributions.

We find that there is no clear relation between the number of groups, and the resulting variances; in particular, this suggests that if the end goal is to obtain a resource-efficient estimator of expectation values, then little, if anything, is gained by spending effort to find smaller clique covers.

\begin{table}
    \begin{tabular}{l||r|r|r|r|r}
        System & ILP & G-SC & CG & RLF & LDF \\
        \hline
        \ce{H2} (STO-3G) & \textbf{0.424} & \textbf{0.424} & \textbf{0.424} & \textbf{0.424} & \textbf{0.424} \\
        \ce{H2} (6-31G) & 5.99 & \textbf{4.81} & 5.38 & 5.73 & 5.48 \\
        \ce{LiH} & \textbf{3.86} & 7.24 & 5.99 & 6.63 & 7.2 \\
        \ce{BeH2} & 19.9 & 23.5 & \textbf{17.4} & 21.1 & 23.2 \\
        \ce{H2O} & 143 & 164 & 126 & \textbf{112} & 131 \\
        \ce{NH3} & &&& 173 & \textbf{172} \\
    \end{tabular}
    \bigskip
    \caption{The variances [Ha$^2$] of the estimators obtained from applying the various grouping algorithms on the \emph{qubit-wise commutation graph}.}
    \label{tab:grouping-algorithm-variance-qwc}
\end{table}

\begin{table}
    \begin{tabular}{l||r|r|r|r|r}
        System & ILP & G-SC & CG & RLF & LDF \\
        \hline
        \ce{H2} (STO-3G) & \textbf{0.288} & \textbf{0.288} & \textbf{0.288} & \textbf{0.288} & \textbf{0.288} \\
        \ce{H2} (6-31G) & 14.8 & 9.56 & 3.93 & \textbf{3.69} & 13.9 \\
        \ce{LiH} &  & & \textbf{6.43} & 17.6 & 9.67 \\
        \ce{BeH2} & & & 18.1 & 36.7 & \textbf{15.1} \\
        \ce{H2O} & & & \textbf{104} & 275 & 291 \\
        \ce{NH3} & & & & 294 & \textbf{145} \\
    \end{tabular}
    \bigskip
    \caption{The variances [Ha$^2$] of the estimators obtained from applying the various grouping algorithms on the \emph{full commutation graph}.}
    \label{tab:grouping-algorithm-variance-fc}
\end{table}

\subsection{Comparison with deterministic schedules}
\label{sec:comparison-with-deterministic}
So far, we have mostly been concerned with Horvitz--Thompson estimators. Other algorithms for estimation of expectation values, such as those obtained from derandomized shadow \cite{Huang2021-ga} and ShadowGrouping \cite{Gresch2025-zy} do not output randomized estimators but rather, for a fixed number of measurements $M > 0$, what we could call a deterministic estimator $\estO^\mathrm{det,M}$ in the form of $M$ fixed measurements to be performed. As such, while the underlying probabilistic nature of quantum states will still mean that these estimators have non-vanishing variance, there is no additional variance from working with a probability distribution of measurements. When these fixed measurements together cover all Pauli strings of the observable of interest, the estimator becomes unbiased. Specifically, let $M > 0$ be a number of measurements of a collection of groups $\mathcal{G}$, let $M_G > 0$ be the number of measurements of a group $G \in \mathcal{G}$, for each Pauli $P \in \mathcal{P}$ let $\mathcal{G}_P \subseteq \mathcal{G}$ denote the set of groups containing $P$, let $M_P = \sum_{G \in \mathcal G_P} M_G$ denote the number of measurements of $P$, which we assume to satisfy $M_P > 0$, and for each $G \in \mathcal{G}_P$, let $n_{P,G}^\pm$ be the number of times the value $\pm1$ has been measured for $P$ following a measurement of $G$, and let $n_P^\pm = \sum_{G \in \mathcal G_P} n^\pm_{P,G}$. Then,
\[
    \hat{\mu}^\mathrm{det}_P = \sum_{G \in \mathcal G_P} \frac{M_G}{M_P} \frac{n_{P,G}^+ - n_{P,G}^-}{M_G} = \sum_{G \in \mathcal G_P} \frac{n_{P,G}^+ - n_{P,G}^-}{M_P} = \frac{n^+_P - n^-_P}{M_P}
\]
is an unbiased estimator of $\trace(\psi P)$, and
\[
    \estO^\mathrm{det,M} = \sum_P c_P\hat{\mu}^\mathrm{det}_P
\]
is an unbiased estimator of $\trace(\psi O)$.

For $P, Q \in \mathcal P$, let $M_{P,Q}$ denote the number of groups $G \in \mathcal{G}$ that contain both $P$ and $Q$ so that in particular, $M_{P, P} = \lvert \mathcal G_P \rvert$. Then the variance of the above estimator is
\begin{align}
    \label{eq:det-variance}
    \Var(\estO^\mathrm{det,M}) = \sum_{P, Q \in \mathcal P} \frac{c_P c_Q M_{P,Q}}{M_P M_Q} \Cov(X_P, X_Q) = \sum_{G \in \mathcal G} M_G \sum_{P, Q \in G} \frac{c_Pc_Q}{M_PM_Q} \Cov(X_P, X_Q),
\end{align}
where as before, $\Cov(X_P, X_Q) = \trace(\psi PQ) - \trace(\psi P) \trace(\psi Q)$ is the covariance between the $\pm 1$-valued random variables defined by $P$ and $Q$.

The expression $M \Var(\estO^\mathrm{det,M})$ can be directly compared to the one-shot variance of the Horvitz--Thompson estimators considered previously; such a comparison is apples-to-apples in the sense that it captures the quality of estimation for the same number of measurements. Indeed, $\estO^\mathrm{det}_P$ is the $M$-experiment Horvitz--Thompson estimator with probability distribution given by $\pi_G = M_G / M$, conditioned on the $M$ outcomes of sampling measurements being exactly $\mathcal{G}$.

The next result shows that the variance of the unconditioned estimator based on this probability distribution will always be at least the variance of the deterministic estimator.

\begin{prop}
    \label{prop:uniform-horvitz-thompson}
    Let $\estO^\mathrm{det,M}$ be an estimator based on a choice of $M > 0$ groups, let $\pi$ be the probability distribution given by counting those groups, and let $\estO_\pi$ be the corresponding Horvitz--Thompson estimator. Then
    \[
        \Var(\estO^\mathrm{det,M}) \leq \frac{1}{M} \Var(\estO_\pi).
    \]
\end{prop}
\begin{proof}
    Since $\pi_G = M_G / M$, and $\pi_P = M_P/M$,
    \[
        \frac{1}{M} \Var(\estO_\pi) = \sum_{G \in \mathcal G} M_G \sum_{P,Q \in G}  \frac{c_Pc_Q}{M_PM_Q} (\Sigma_{PQ} + \mu_P\mu_Q) - \frac{\trace(\psi O)^2}{M}.
    \]
    By reordering the two sums and using that $M_{P, Q} = \sum_{G \in \mathcal G_P \cap \mathcal G_Q} M_G$, it follows that
    \begin{align*}
        \frac{1}{M} \Var(\estO_\pi) - \Var(\estO^\mathrm{det,M}) &= \sum_{G \in \mathcal G} M_G \sum_{P, Q \in G} \frac{c_P c_Q \mu_P \mu_Q}{M_P M_Q} - \frac{\trace(\psi O)^2}{M}\\
         &= \sum_{G \in \mathcal G} M_G \left(\sum_{P \in G} \frac{c_P \mu_P}{M_P} \right)^2 - \frac{\trace(\psi O)^2}{M}.
    \end{align*}
    By Cauchy--Schwarz,
    \[
        \trace(\psi O)^2 = \left(\sum_{G \in \mathcal G} \pi_G \sum_{P \in G} \frac{c_P \mu_P}{\pi_P}\right)^2 \leq 
        \left(\sum_{G \in \mathcal G}(\sqrt{\pi_G})^2 \right) \sum_{G \in \mathcal G} \left( \sqrt{ \pi_G} \sum_{P \in G} \frac{c_P \mu_P}{\pi_P}\right)^2 = \sum_{G \in \mathcal{G}} M_G \left(\sum_{P \in G} \frac{c_P \mu_P}{M_P}\right)^2.
    \]
    
\end{proof}
In Table~\ref{tab:compare-det-random}, we provide a comparison of estimator variances for measurement schedules obtained from ShadowGrouping \cite{Gresch2025-zy}. We take $M$ to be the minimal number such that the estimators are unbiased and compare the variance of the deterministic estimator $\estO^\mathrm{det,M}$, scaled by $M$, to those obtained by using the same measurements in a randomized setting, using either the group counting distribution $\pi_{\mathrm{uniform}}$ as in Proposition~\ref{prop:uniform-horvitz-thompson}, or one optimized with SLSQP, denoted $\pi_{\mathrm{opt}}$, as in Section~\ref{sec:estimation-overlapping-groups}. As the results in the table show, the difference between the variances can be rather dramatic; as such, some care must be taken when comparing randomized schemes to deterministic ones.

\begin{table}
    \centering
    \begin{tabular}{l|r||r|r|r}
        System &$M$& $M\Var\left(\estO^\mathrm{det,M}\right)$ & $\Var\left(\estO_{\pi_{\mathrm{uniform}}}\right)$ & $\Var\left(\estO_{\pi_{\mathrm{opt}}}\right)$ \\
        \hline
        \ce{H2} (STO-3G) & 5 & 0.195 & 4.17 & 0.424 \\
        \ce{H2} (6-31G) & 44 & 2.06 & 8.62 & 5.57 \\
        \ce{LiH} & 142 & 1.70 & 6.76 & 3.74 \\
        \ce{BeH2} & 161 & 6.65 & 27.7 & 16.5 \\
        \ce{H2O} & 239 & 19.3 & 301.0 & 77.5 \\
        \ce{NH3} & 653 & 45.9 & 191.0 & 88.6 \\
    \end{tabular}
    \bigskip
    \caption{Estimator variances for measurement schedules obtained from ShadowGrouping. All entries are in units of [Ha$^2$].}
    \label{tab:compare-det-random}
\end{table}

\subsection{Derandomization and allocation}
\label{sec:derandomization-and-allocation}
In the above, we considered the process of taking a deterministic schedule, such as what is produced by ShadowGrouping, and ``randomizing'' it by drawing the measurements from a distribution as opposed to running through them in order. In the simplest example, where we would use $\pi_{\mathrm{uniform}}$ to sample, hence ensuring that asymptotically, the fraction of the measurement budget allocated to any given measurement would agree with that of the deterministic estimator, we see how the additional variance from sampling measurements can be significant. Indeed, the deterministic estimator may be seen as the result of conditioning the randomized estimator on a particular outcome of the measurement-choosing random variable $I$, so for a fair comparison between the two classes of estimators, it becomes natural to evaluate the variance of the Horvitz--Thompson estimators, conditioned on particular sampling outcomes.

\subsubsection{Optimal allocation}
Just like we could find probability distributions $\pi$ for Horvitz--Thompson estimators so as to minimize their variance, we can ask how to choose $M_G > 0$, $\sum_G M_G = M$, for a deterministic estimator in a way that minimizes its variance \eqref{eq:det-variance}. This is generally a hard integer optimization problem, but as long as $M$ is large, we may approximate its solution by assuming that $M_G = \pi_G M$ for $\pi_G \in (0, 1]$, $\sum_G \pi_G = 1$; at the end of the day, this means that some rounding will have to be performed to ensure that $M_G$ are integers, but when $M$ is large, this rounding has a correspondingly smaller impact on the resulting variance. Let $\pi_P = \sum_{G \in \calG_P} \pi_G$ and $\pi_{P,Q} = \sum_{G \in \calG_P \cap \calG_Q} \pi_G$. We are then led to choose $\pi_G$ in a way that minimizes
\begin{align}
    \label{eq:opt-alloc}
    \sum_{P, Q \in \calP} \frac{c_P c_Q \pi_{P,Q}}{\pi_P \pi_Q} \Cov(X_P, X_Q) = \sum_{G \in \calG} \pi_G \sum_{P, Q \in G} \frac{c_P c_Q}{\pi_P\pi_Q} \Cov(X_P, X_Q).
\end{align}

\begin{example}
    \label{ex:min-var-example}
    Before moving on, let us consider an example where optimal allocation for a determinsitic estimator looks qualitatively different than the optimal assignment of probabilities for Horvitz--Thompson estimators. Specifically, let $X_1$ and $X_2$ be independent random variables with means $\mu_1 = 0$ and $\mu_2 = \tfrac{1}{2}$, and variances $\Var(X_1) = 1$, and $\Var(X_2) = \tfrac{3}{4}$. Assume that each variable belongs to its own group, $G_1$ and $G_2$ respectively, and suppose that we are tasked with estimating $\bbE[X_1 + \sqrt{\tfrac{4}{3}}X_2]$.

    Then the optimization objective \eqref{eq:opt-alloc} becomes
    \[
        \frac{1}{\pi_{G_1}} \Var(X_1) + \frac{4}{3\pi_{G_2}} \Var(X_2) = \frac{1}{\pi_{G_1}} + \frac{1}{\pi_{G_2}},
    \]
    which is minimized when $\pi_{G_1} = \pi_{G_2} = \tfrac{1}{2}$, i.e. we should allocate measurements uniformly.

    On the other hand, for the Horvitz--Thompson estimator, by \eqref{eq:opt-non-overlapping}, the optimal assignment of probabilities is
    \[
        \pi_{G_1} = \frac{1}{1+\sqrt{\tfrac{4}{3}}}, \quad \pi_{G_2} = \frac{\sqrt{\tfrac{4}{3}}}{1+\sqrt{\tfrac{4}{3}}},
    \]
    i.e. that we should sample $X_2$ with higher frequency than $X_1$. See Figure~\ref{fig:min-var-example-plot}.
\end{example}
\begin{figure}
    \centering
    \includegraphics[width=0.75\linewidth]{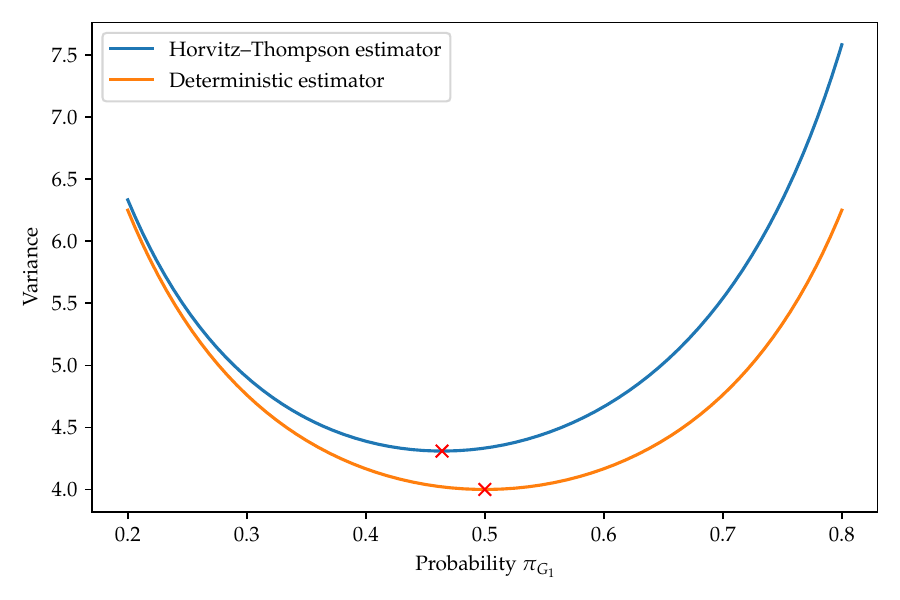}
    \caption{Variances of the estimators from Example~\ref{ex:min-var-example}.}
    \label{fig:min-var-example-plot}
\end{figure}

Coming back to the general optimization problem, as in the randomized case, we may approximate $\Cov(X_P, X_Q)$ by $c \delta_{P,Q}$ for some constant $c > 0$, so we are left with having to optimize
\[
    \sum_{G \in \calG} \pi_G \sum_{P \in G} \frac{c_P^2}{\pi_P^2}.
\]
This is exactly the same optimization problem as \eqref{eq:opt-problem-overlapping-groups}, so we may use the same approaches to solve it.

To benchmark this, we consider a few variants. In either case, assume that a grouping, i.e. a clique cover, $\calG$ has already been produced.
\begin{enumerate}
    \item Let $M = \lvert \calG \rvert$ and allocate one measurement to each group. This produces an $M$-trial deterministic estimator $\estO^{\mathrm{det},M}_{\pi_\mathrm{uniform}}$.
    \item Assume that a probability distribution, $\pi$, on $\calG$ is available. Let $\tilde{M} = 3\lvert \calP\rvert > \lvert \calG \rvert$, and allocate measurements by letting $M_G = \lceil\pi_G \tilde{M} \rceil$, and let $M = \sum_{G \in \calG} M_G \geq \tilde{M}$. This defines another deterministic estimator than we denote $\estO^{\mathrm{det},M}_\pi$.
\end{enumerate}

\subsubsection{Results}
\label{sec:allocation-results}

In Table~\ref{tab:compare-lvf-several-deterministic}, we compare the variances of the resulting deterministic estimators to that of the randomized estimator where $\calG$ is obtained through LVF grouping on the QWC and FC graphs, in both cases post-processed with LDF maximalization, and for which the distribution $\pi$ on $\calG$ is either $\pi_{\ell^2}$ given by \eqref{eq:opt-non-overlapping}, or $\pi_{\mathrm{opt}}$ obtained by using SLSQP to solve \eqref{eq:opt-problem-overlapping-groups}.

In an attempt to quantify the loss of accuracy in having to assume a prior on the a priori unavailable quantum covariance, note that for the small benchmark molecules, all covariances can be calculated quickly, so even though this will likely not be possible for larger molecules, we can use this information during the optimization of $\pi$, so as to understand how useful having exact covariances values would be.

Concretely, for a given grouping $\calG$, assume that $\Cov(X_P, X_Q)$ is known for all $P, Q \in G$ and all $G \in \calG$. We may then try to find $\pi$ minimizing \eqref{eq:opt-problem-overlapping-groups}. In its general form, this is still a hard problem, so we reduce the level of ambition and still assume that $\Sigma_{PQ} = 0$ for $P \not= Q$. In other words, we solve the same optimization problem as before, but this time we use the actual variances of each Pauli string, and we choose $\pi$ such that
\[
    \sum_{G \in \calG} \pi_G \sum_{P \in G} \frac{c_P^2}{\pi_P^2} \Var(X_P)
\]
is minimized, as before relying on an off-the-shelf SLSQP solver. We call the resulting estimator $\estO^{\mathrm{det},M}_{\pi_{\text{known var.}}}$.

\begin{table}[h!]
    \centering
    \begin{tabular}{l||r|r|r|r||r|r|r|r}
        \multicolumn{1}{c||}{} & 
        \multicolumn{4}{c||}{QWC} & 
        \multicolumn{4}{c}{FC} \\
        System & $\pi_{\mathrm{uniform}}$ & $\pi_{\ell^2}$ & $\pi_{\mathrm{opt}}$ & $\pi_{\text{known var.}}$ & $\pi_{\mathrm{uniform}}$ & $\pi_{\ell^2}$ & $\pi_{\mathrm{opt}} $  & $\pi_{\text{known var.}}$ \\
        \hline
        \ce{H2} (STO-3G) & 0.195 & 0.157 & 0.157 & 0.137 & 0.125 & 0.128 & 0.206 & 0.125 \\
        \ce{H2} (6-31G) & 2.2 & 3.26 & 2.96 & 1.85 & 0.964 & 0.954 & 1.27 & 0.815 \\
        \ce{LiH} & 3.57 & 2.45 & 2.26 & 1.31 & 1.64 & 1.68 & 2.13 & 1.21 \\
        \ce{BeH2} & 10.4 & 7.27 & 7.25 & 4.14 & 3.18 & 3.41 & 4.71 & 2.39 \\
        \ce{H2O} & 64.9 & 30.2 & 52.0 & 19.1 & 21.0 & 12.8 & 31.4 & 8.89 \\
        \ce{NH3} & 153 & 76.3 & 77.3 & 37.0 & 30.4 & 24.1 & 34.1 & 13.2 \\
    \end{tabular}
    \bigskip
    \caption{Comparison of $M \Var\left(\estO^{\mathrm{det},M}_\pi\right)$ for the four choices of $\pi$. All entries are in units of [Ha$^2$].}
    \label{tab:compare-lvf-several-deterministic}
\end{table}

Perhaps surprisingly, the additional information from the optimized probability distribution $\pi$ does not generally appear to be helpful; that is, the results for $\pi_{\mathrm{opt}}$ are not in general better than those for $\pi_{\mathrm{uniform}}$, but as expected, including the actual variance always yields the best results. This suggests that the approximation from assuming constant variance across all Pauli strings is not useful: If we use $c_P^2 \Var(X_P)$ instead of just $c_P^2$ for our optimization we get much better results. This shows that there is high value in being able to estimate $\Var(X_P)$ a priori; an extreme case that we have already used implicitly is that $\Var(X_I) = 0$ which was motivation that the identity Pauli string should be discarded entirely. In Figure~\ref{fig:coeff-vs-var-plot}, we indicate the relation between $c_P^2$ and $c_P^2\Var(X_P)$; in each case, the points with large coefficients correspond to Pauli strings that are the identity with one term replaced by a $Z$.

\begin{figure}
    \centering
    \includegraphics[width=0.99\linewidth]{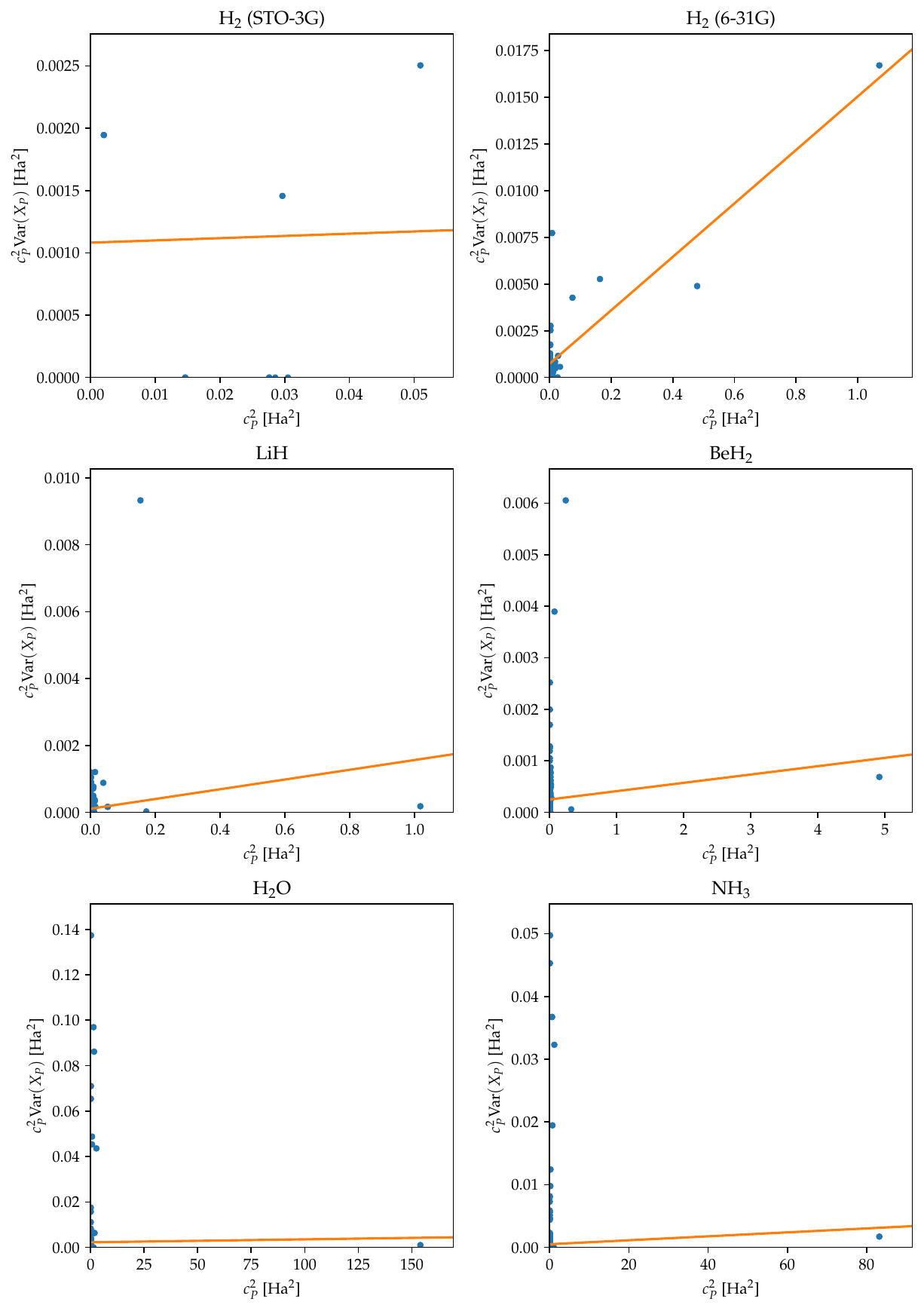}
    \caption{The relation between $c_P^2$ and $c_P^2\Var(X_P)$ for the 6 benchmark Hamiltonians. Each plot shows the points $(c_P^2, c_P^2\Var(X_P))$ for all Pauli strings $P$ for that Hamiltonian. The line is a linear least-squares fit.}
    \label{fig:coeff-vs-var-plot}
\end{figure}
\section{Bias--variance tradeoff}
\label{sec:bias-variance}
Up until this point, we have been concerned with minimizing the variance of an unbiased estimator $\estO$. For practical purposes, it may be beneficial to allow the estimator to be biased, allowing for a smaller variance than what is otherwise possible; another meaningful objective would be to minimize the mean squared error
\[
    \MSE(\estO) = \bbE[(\estO - \expecO)^2] = \Var(\estO)  + \Bias(\estO)^2
\]
or, as is particularly common in the classical shadow literature, to minimize the number of measurements needed to ensure that $\Pr(\lvert \estO - \expecO\rvert < \eps) < 1 - \delta$ for a given desired error $\eps > 0$ and failure probability $\delta > 0$.

And indeed, in e.g. \cite{Wu2023-sw, Li2025-vm} the bias-variance tradeoff features explicitly, but it may also feature implicitly in work based on deterministic estimators, such as \cite{Huang2021-ga}, whenever the measurement count is too low to produce a clique cover.

Here, we briefly study the scale of the bias that may appear in practice and discuss some pitfalls when comparing the quality of unbiased and biased estimators.

\subsection{Introducing bias by Pauli sum truncation}
One way to introduce bias is to simply remove from $\calP$ the Pauli strings which may be deemed irrelevant a priori. Decomposing $\calP = \calP' \sqcup \calP''$, with a corresponding decomposition $O = O' + O''$ of the observable, an unbiased estimator of $O'$ also defines a biased estimator $O$ whose bias is $\langle O'' \rangle$.

Such a truncation may be included systematically, \cite[App.~D]{Gresch2025-zy}, \cite{McClean2016-er}, and may be informed by the problem at hand and the nature if the quantum state; an extreme case is \cite{Majland2023-lz} in which only a single measurement circuit survives, which can be sufficient when limited precision is required, or for low-correlation states.

\begin{table}
    \centering
    \begin{tabular}{l||r|r|r}
        System & $E_{\mathrm{HF}} - \langle O \rangle$ & ROGS & Derandomized shadow \\
        \hline
        \ce{H2} (6-31G) & 0.0251 &0.03&0.06 \\
        \ce{LiH} &0.0197&0.02&0.03\\
        \ce{BeH2} &0.0339&0.02&0.06\\
        \ce{H2O} &0.0607&0.09&0.12\\
        \ce{NH3} &0.0770&0.09&0.18
    \end{tabular}
    \caption{Comparison between the energy error obtained by the constant estimator outputting $E_{\mathrm{HF}}$, and estimates of the root mean squared error of two $1000$ measurement estimators. Here, $\expecO = \langle H \rangle$ is the ground state energy}
    \label{tab:compare-hf-energy-rmse}
\end{table}

It may also be included implicitly in derandomization procedures when the measurement count is low. In Table~\ref{tab:compare-hf-energy-rmse}, we show the bias as reported by \cite[Table 2]{Li2025-vm} when running derandomized shadow \cite{Huang2021-ga} and ROGS \cite{Li2025-vm} with $M = 1000$ measurements on the benchmark cases. We also show the difference between the ground state energy, i.e. $\expecO$, and the Hartree--Fock energy, which, in all of the benchmark cases, is the energy $E_{\HF}$ of the state $\lvert 0 \dots 01 \dots 10 \dots0 1 \dots 1\rangle$, the computational basis vector built by using $\tfrac{n_v}{2}$ $0$s, $\tfrac{n_o}{2}$ $1$s, $\tfrac{n_v}{2}$ $0$s, and $\tfrac{n_o}{2}$ $1$s, where $n_v$ and $n_o$ are the number of virtual and occupied spin orbitals respectively. This energy is produced implicitly during the Hartree--Fock method leading to the Hamiltonians at hand, and as such for our purposes can be considered computationally free, in that no additional effort is required to get it. What this also means is that the constant estimator $\estO \equiv E_{\HF}$, whose variance vanishes, has lower bias than the ones estimated for the estimators in the table. This also means that some amount of care is necessary to define good baseline error levels for benchmarking, and for this reason, we will restrict ourselves to considering the variance of unbiased estimators in the below.

For algorithms generating a collection of fixed measurements, the bias of the resulting deterministic estimator vanishes once every Pauli string is measured at least once. In Figure~\ref{fig:shadow-grouping-bias}, we show the biases, standard deviations $\sqrt{\Var(\estO)}$, and root mean square errors $\sqrt{\MSE(\estO)}$ of the deterministic estimators for \ce{NH3} ground state energy estimation obtained from the ShadowGrouping \cite{Gresch2025-zy} method for low measurement counts, comparing them to the bias of the constant Hartree--Fock energy estimator above. Note that in the literature, it is common to report statistics like the root mean square error as estimates obtained from simulation, typically using about $100$ runs, but that for small systems like the benchmark molecules, with a small number of different measurements, we can also calculate these statistics exactly, and all results given here are exact. We may note, for instance, that the after $662$ measurements, every Pauli string has been measured at least once, such that the bias vanishes, yet the root mean square error is still significantly larger than that of the constant Hartree--Fock energy estimator.

\begin{figure}
    \centering
    \includegraphics[width=0.7\linewidth]{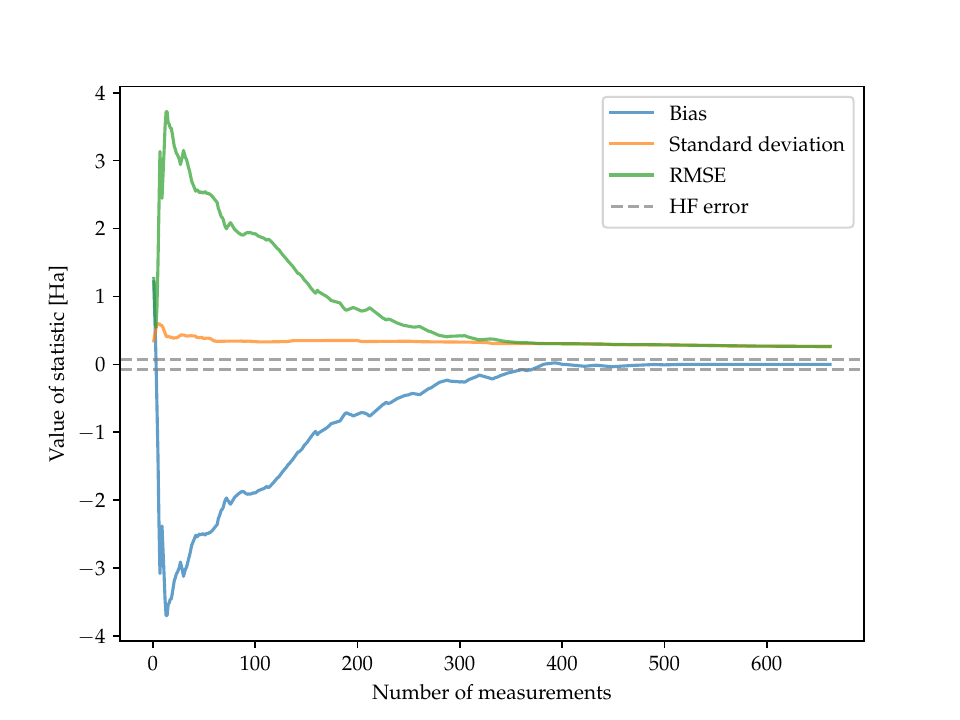}
    \caption{Bias, standard deviation, and RMSE of $\estO^{\mathrm{det},M}$ as obtained from ShadowGrouping for \ce{NH3}, for different values of $M$. The horizontal gray dashed lines are $\pm E_{\mathrm{HF} - \expecO}$.}
    \label{fig:shadow-grouping-bias}
\end{figure}

\subsection{The variance of estimators from derandomization}
\label{sec:variance-derandom}
In approaches such as those of \cite{Gresch2025-zy}, the focus is not so much on minimizing variance as it is on making as much of a given measurement budget as possible, so to be able to make a fair comparison with Horvitz--Thompson estimators, ideally we would have an asymptotic variance to compare to. In Figure~\ref{fig:shadow-grouping-one-shot-var}, we show the value of $M \Var( \estO_M^\mathrm{\det})$ for $M \in \{1, \dots, 3\lvert \calP\rvert\}$, i.e. the equivalent of the one-shot variance for the \ce{LiH} benchmark case; its convergence suggests that a reasonable comparison with other methods would be provided by choosing a suitably large value of $M$. We find that for the benchmark cases, taking $M = 3 \lvert \calP \rvert$ as in \cite{Zhang2023-lr} appears to be reasonable; see Table~\ref{tab:shadow-grouping-one-shot-var}.

\begin{figure}
    \centering
    \includegraphics[width=0.7\linewidth]{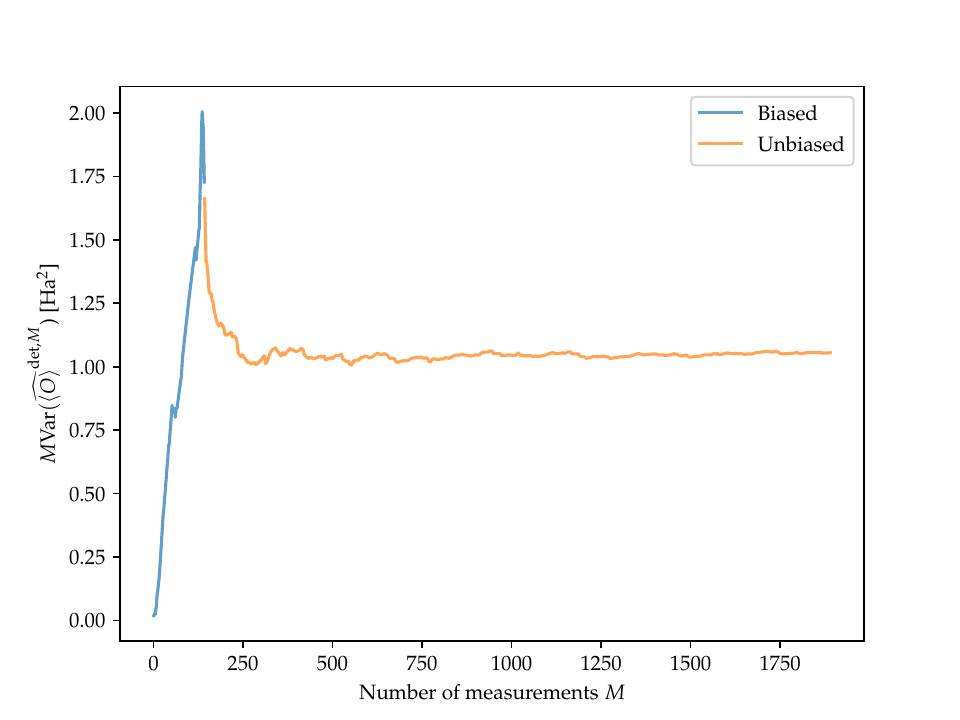}
    \caption{Values of $M\Var(\estO^{\mathrm{det},M})$ for the ShadowGrouping estimator for \ce{LiH} for different values of $M$. The color indicates whether or not the estimator is unbiased.}
    \label{fig:shadow-grouping-one-shot-var}
\end{figure}
\begin{table}
    \centering
    \begin{tabular}{l||r|r|r}
        System & $M_{\mathrm{minimal}}$&$3\lvert \calP \rvert$ & $10\lvert \calP\rvert$ \\
        \hline
        \ce{H2} (STO-3G) & 0.195 & 0.127 & 0.127 \\
        \ce{H2} (6-31G) & 2.06 & 1.77 & 1.72 \\
        \ce{LiH} & 1.7 & 1.06 & 1.06 \\
        \ce{BeH2} & 6.65 & 3.79 & 3.75 \\
        \ce{H2O} & 19.3 & 11.7 & 11.6 \\
        \ce{NH3} & 45.9 & 23.3 & 23.1 \\
    \end{tabular}
    \caption{Values of $M\Var(\estO^{\mathrm{det},M})$ for the ShadowGrouping estimator where $M \in \{M_\mathrm{minimal}, 3\lvert \calP \rvert, 10 \lvert \calP \rvert\}$. Here, $M_\mathrm{minimal} \in \mathbb{N}$ is the smallest number of measurements required such that each Pauli string is measured at least once.}
    \label{tab:shadow-grouping-one-shot-var}
\end{table}

\subsection{The central limit theorem}
\label{sec:clt}
Once we have limited our attention to unbiased estimators, whether Horvitz--Thompson estimators or deterministic ones for sufficiently large measurement counts, we may ask for precise estimates of the amount of resources needed to reach a certain estimation accuracy. As we have noted, in the work based on classical shadow, it is natural to ask if one can bound the number of measurements required to reach a target estimation error with high probability.

In general, for large enough measurement counts, the central limit theorem (CLT) tells us that the estimation error tends to a normal distribution for any of the estimators we have considered, which means that we can derive approximate resource estimates from the variance of those estimators alone.

For the benchmark cases, we can compare the estimates derived from assuming normally distributed errors to those obtained by simulation. Recall that given $n$ independent samples of a normally distributed variable $X \sim N(\mu, \sigma^2)$,
\[
    \Pr(\lvert \bar{X}_n - \mu \rvert < \eps) = 2\Phi\left(\frac{\eps}{\sqrt{n}}{\sigma}\right) - 1,
\]
where $\Phi$ is the standard normal CDF, so that in particular, we need
\[
    n = \left\lceil \frac{\sigma^2 z_{1-\delta/2}^2}{\eps^2}\right\rceil
\]
samples to reach a desired failure probability $\delta$, where here $z_{1-\delta/2} = \Phi^{-1}(1-\delta/2)$.

In Table~\ref{tab:clt-approx}, we use the CLT approximation to estimate the number of measurements $M$ required to reach an error rate of $\epsilon = 0.0016\,\Ha$ (or about $1\,\text{kcal/mol}$, often referred to as ``chemical accuracy'') with a failure probability of $\delta = 0.05$. We compare this to a simulation of an estimator at the same number of measurements, run a total of 1,000 times. Similarly, Figure~\ref{fig:qq-plot} shows Q--Q plots of the resulting estimates against the normal distribution with a location given by the theoretical ground state energy, and variance $\tfrac{1}{M}\Var(\estO_\pi)$. An upshot is that for these cases, the variance itself is sufficient to produce resource estimates.

\begin{table}
    \centering
    \begin{tabular}{l||r|r|r}
        System & $\Var(\estO_\pi)$ & $M$ & Success rate\\
        \hline
        \ce{H2} (STO-3G) & 0.288& 432,175  &94.9\,\% \\
        \ce{H2} (6-31G) &13.9 & 20,789,187 & 95.5\,\%\\ 
        \ce{LiH} & 9.67 & 14,514,273 & 94.4\,\%\\
        \ce{BeH2} & 15.1 & 22,648,509 & 95.1\,\% \\
        \ce{H2O} & 291 & 437,294,227 & 94.8\,\%\\
        \ce{NH3} & 145 & 217,446,003 & 96.0\,\%\\
    \end{tabular}
    \caption{For every molecule, the variance of the Horvitz--Thompson estimator $\estO_\pi$ based on LDF grouping on the FC graph, with LDF maximalization and with probabilities optimized using SLSQP. Based on that variance, the number $M$ of measurements required for a success rate of $95\,\%$, and the actual success rate from 1,000 simulations.}
    \label{tab:clt-approx}
\end{table}

\begin{figure}[H]
    \centering
    
    \includegraphics[width=0.9\linewidth]{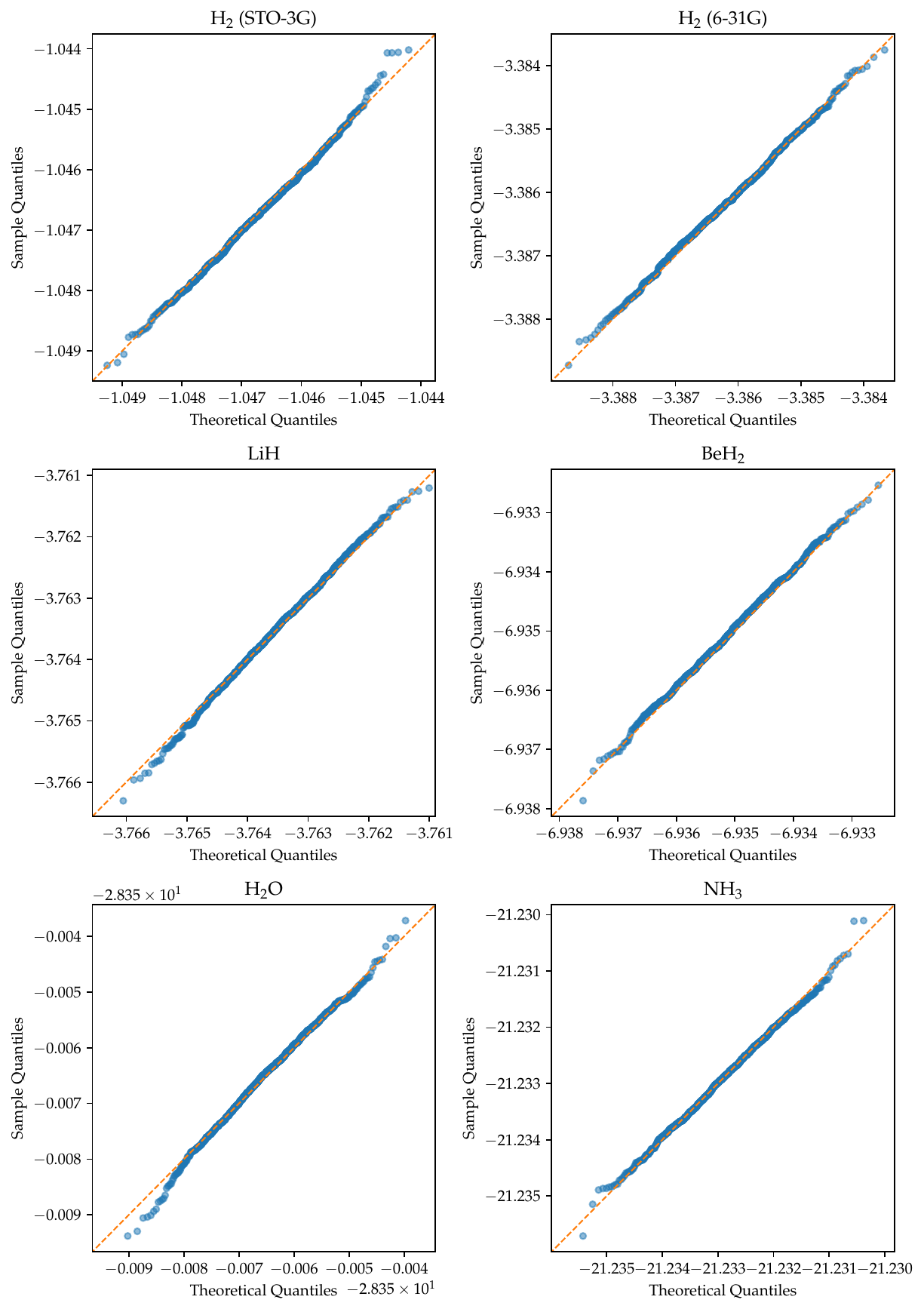}
    \caption{Q--Q plot of 1,000 estimates from the Horvitz--Thompson estimator $\estO_\pi$ based on LDF grouping on the FC graph, with LDF maximalization and with probabilities optimized using SLSQP, after $M$ measurements as in Table~\ref{tab:clt-approx}, the normal distribution centered at the true ground state energy and with variance $\tfrac{1}{M}\Var(\estO_\pi)$.}
    \label{fig:qq-plot}
\end{figure}

Along the same lines, we present numerical experiments that evaluate the success rates as the number of measurements $M$ increases. We consider two scenarios: one using the Horvitz–Thompson estimator and another using the deterministic estimator. In each case, we begin by fixing a molecular structure -- specifically, LiH -- and performing the numerical experiment with various grouping and post-processing strategies. These are defined by selecting cases with different enough variances that they can be clearly displayed together. After identifying the best-performing combination, we apply it to all molecules in the benchmarking set. 

The success rates in these figures are again defined by counting the number of estimates where the error with respect to the exactly calculated ground state energy is within chemical accuracy (i.e. less than $0.0016\,\Ha$). The actual data displayed in the plots is obtained by repeating the experiment 1,000 times, and then performing a bootstrapping procedure, where we resample the distribution to calculate the mean success rates and their standard deviation.

We also display the curves obtained by assuming the CLT approximation and sampling a normal distribution with the variance corresponding to the molecule and method of choice. In Figure~\ref{fig:ht-exp-plot} we show the results of the experiment using the Horvitz--Thompson estimator (and a few choices of sampling strategies), while the results using the deterministic estimator (and uniform sampling) are shown in Figure~\ref{fig:det-exp-plot}.

All combinations of molecules, grouping, and postprocessing schemes shown in these figures are also summarized in Tables~\ref{tab:var-comparison-horvitz-thompson} and \ref{tab:var-comparison-deterministic-uniform}. One exception is the "Trivial" result shown in Figure~\ref{fig:ht-exp-plot}, which assumes the extreme case of no grouping, where each Pauli string is measured independently.
\begin{figure}[p]
    \centering
    \hspace*{-.2cm}
    \includegraphics[width=0.8\linewidth]{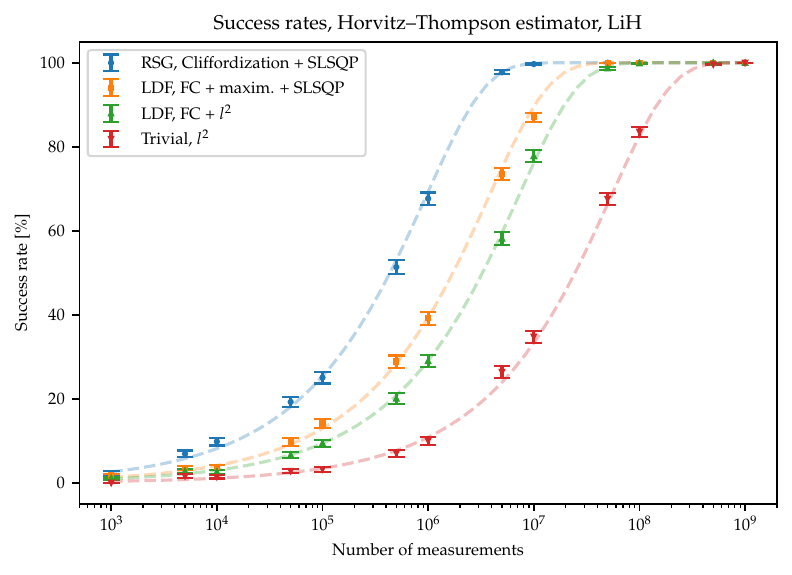}
    \includegraphics[width=0.8\linewidth]{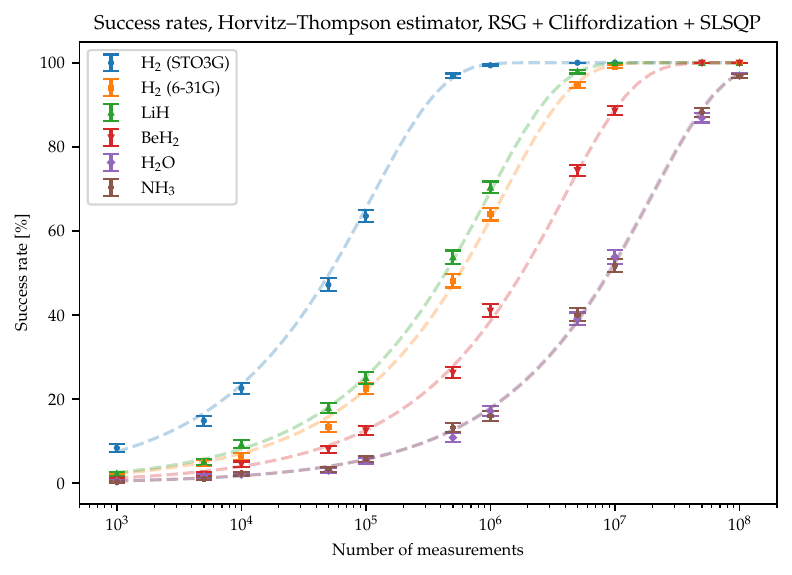}
    \caption{Success rates in the estimation of ground state energies, using the Horvitz--Thompson estimator, for different grouping methods and molecular Hamiltonians and increasingly larger number of measurements $M$. The top figure shows the success rates obtained for a fixed molecule (LiH) and a selection of different grouping schemes (see Table~\ref{tab:var-comparison-horvitz-thompson} for all considered combinations). The bottom figure displays the results using the best-performing method in general (randomized ShadowGrouping (RGS) + Cliffordization + SLSQP) for all molecules of the standard benchmarking set. The dashed curves in each plot correspond to assuming the CLT approximation and sampling from a normal distribution with the corresponding variance given by the estimator (in the bottom plot, the dashed curves for \ce{H2O} and \ce{NH3} overlap).}
    \label{fig:ht-exp-plot}
\end{figure}

\begin{figure}[p]
    \centering
    \includegraphics[width=0.8\linewidth]{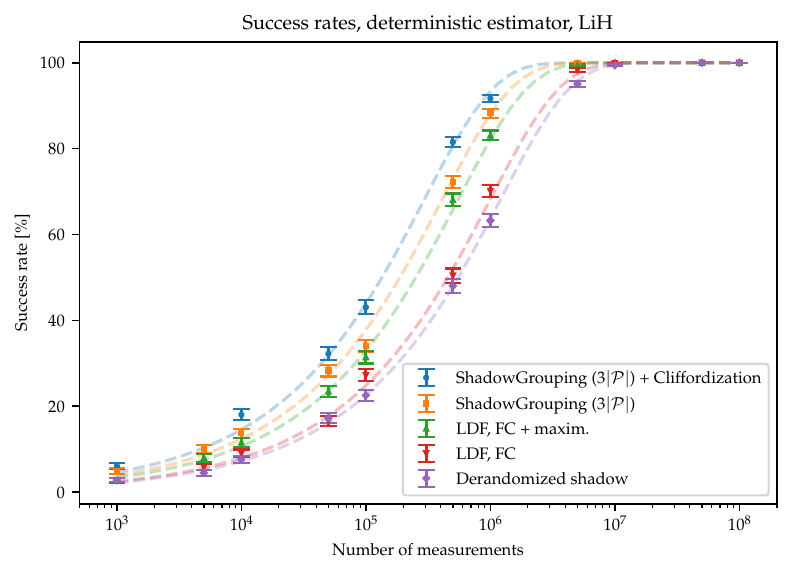}
    \includegraphics[width=0.8\linewidth]{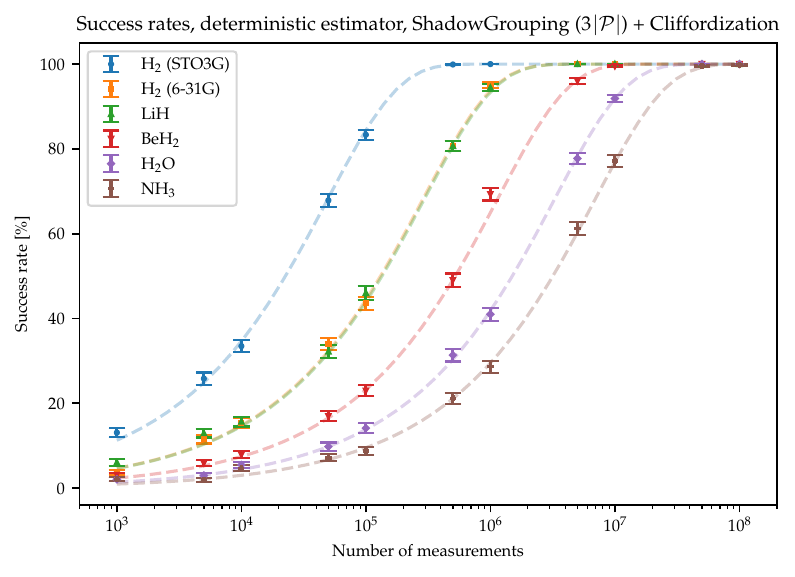}

    \caption{Success rates in the estimation of ground state energies, using the deterministic estimator, for different grouping methods and molecular Hamiltonians and increasingly larger number of measurements $M$. In the bottom plot, the dashed curves for \ce{H2} (6-31G) and LiH overlap.}    
    \label{fig:det-exp-plot}
\end{figure}
\clearpage
\section{Comparison with existing methods}
\label{sec:comparison-with-existing}

At this point we have everything we need to provide a comparison between the different approaches to constructing estimators, to investigate the usefulness of the various augmentations of randomized measurement schemes and how they compare to their deterministic counterparts.

\subsection{Overview of algorithms}
\label{sec:comparison-algorithms}
We will compare the variances of a number of different estimators; as the details on how they have been put on the same footing are significant, we first describe each of them in detail.

Firstly, we take a number of results from the literature:
\begin{itemize}
    \item First, $\ell^1$ sampling refers to the simple ``single Pauli string per measurement'' estimator from Example~\ref{ex:horvitz-thompson}, with probabilities $\pi_P$ proportional to $\lvert c_P\rvert$. The values are taken from \cite{Hadfield2021-hx}.
    \item $\ell^1$ LDF grouping is the result of using the QWC version of LDF grouping, and using the Horvitz--Thompson estimator in which $\pi_G$ is proportional to $\sum_{P \in G} \lvert c_P \rvert$. The results are from \cite{Hadfield2021-hx} but as noted have been reproduced in Table \ref{tab:compare-l1-and-l2} above.
    \item Classical shadow refers to the work of \cite{Huang2020-ws}. Results are taken from \cite{Hadfield2021-hx}.
    \item LBCS is locally-biased classical shadow, as described in \cite{Huang2020-ws} from which the results have also been lifted.
    \item Fermionic shadow (\cite{Zhao2021-gn}) is a variant of classical shadow based on sampling fermionic Gaussian unitaries (FGUs) as opposed to, say, Pauli strings or Clifford operators.
    \item Basis rotation grouping (variance estimates from \cite{Zhao2021-gn}, method from \cite{Huggins2021-al}) refers to estimates arising from using low-rank factorizations of the coefficient tensor to partition the Hamiltonian.
    \item Joint fermionic grouping (\cite{Majsak2025-av}) is another variation on sampling FGUs while simultaneously measuring multiple non-commuting observables. The results are from \cite[Table 1]{Majsak2025-av}, scaled by a factor of $4$, as the table caption suggests.
    \item Uniform (derandomized shadow, OGM, shadow grouping) refers to versions of derandomized shadow \cite{Huang2021-ga}, overlapped grouping measurement \cite{Wu2023-sw} and ShadowGrouping \cite{Gresch2025-zy} that have been (re-)randomized with uniform probability distributions as described in \cite{Zhang2023-lr}. The results are from \cite[Table~I]{Zhang2023-lr}.
    \item C-LBCS is a composite measurement scheme obtained by composing several locally-biased classical shadows, as described in \cite{Zhang2023-lr} from which the results have also been taken.
\end{itemize}
We supplement these with a collection of new results:
\begin{itemize}
    \item LDF, LDVF, and LVF are the greedy graph colouring based algorithms described previously. For each of them we consider qubit-wise (QWC) and full commutation (FC) versions, and for each of those we consider both variants in which the clique covers are kept as partitions, with Horvitz--Thompson probabilities defined by the $\ell^2$ norm as in \eqref{eq:opt-non-overlapping}, as well as versions for which we have applied LDF maximalization within the approporiate graphs, together with SLSQP optimization of the Horvitz--Thompson probabilities.
    \item Derandomized shadow refers to the algorithm of \cite{Huang2021-ga}\footnote{We use an implementation that matches the reference implementation from \url{https://github.com/hsinyuan-huang/predicting-quantum-properties/}}. We use the unweighted version of the algorithm, run until every Pauli string is covered by at least one measurement. Note that derandomized shadow also comes in a version that makes use of the coefficients of the Hamiltonian; we find that by running that one until every Pauli string is covered, one ends up with a large number of pure $X$-basis measurements leading to detrimental overall variance. This is likely a simple implementation bug, but we leave out the results. As derandomized shadow produces Pauli measurements, we also give the result of applying Cliffordization, as described in Section~\ref{sec:cliffordization}.
    \item RDS (randomized derandomized shadow) is the result of using the measurement circuits produced by derandomized shadow as above, then treating those as the groups of a randomized measurement scheme with a corresponding Horvitz--Thompson estimator for which the probabilities are obtained from SLSQP optimization. As above, the resulting cliques are not in general maximal and so we also consider the result of applying LDF maximalization in both the QWC and FC (Cliffordization) cases on top.
    \item Deterministic ShadowGrouping refers to the algorithm of \cite{Gresch2025-zy}\footnote{We use the implementation from \url{https://gitlab.com/GreschAI/shadowgrouping} -- note that this implementation only includes a QWC version, although \cite{Gresch2025-zy} does discuss which modifications would be required to allow full commutation.}, run until all Pauli strings are covered (``minimal'') or until $3 \lvert \calP\rvert$ measurement circuits have been produced (as justified in Section~\ref{sec:variance-derandom}). In each case, we apply Cliffordization to the result.
    \item RSG (randomized ShadowGrouping) is the result of running ShadowGrouping until all Paulis are covered, then using the resulting groups as the basis for a randomized scheme, as with RDS.
\end{itemize}

\subsection{Results}
\label{sec:comparison-results}
In Tables~\ref{tab:var-comparison-existing}--\ref{tab:var-comparison-deterministic-l2}, we show the results of applying all of the above algorithms to the $6$ benchmark observables, recording the resulting variances of the resulting estimators. Concretely,
\begin{itemize}
    \item Table~\ref{tab:var-comparison-existing} lists a number of results as obtained from the literature,
    \item Table~\ref{tab:var-comparison-horvitz-thompson} lists the variances for several Horvitz--Thompson estimators
    \item Table~\ref{tab:var-comparison-deterministic-uniform} lists the variances obtained from deterministic estimators with uniform measurement allocation, and finally,
    \item Table~\ref{tab:var-comparison-deterministic-l2} lists the variances obtained from deterministic estimators with $\ell^2$ allocation, as in Section~\ref{sec:allocation-results}.
\end{itemize}

\begin{table}
\begin{tabular}{l|r|r|r|r|r|r}
Method & \ce{H2} (STO-3G) & \ce{H2} (6-31G) & \ce{LiH} & \ce{BeH2} & \ce{H2O} & \ce{NH3} \\ \hline
$\ell^1$ sampling & 2.49 & 120 & 138 & 418 & 4360 & 3930 \\
$\ell^1$ LDF grouping & 0.402 & 22.3 & 54.2 & 135 & 1040 & 891 \\
classical shadow & 1.97 & 51.4 & 266 & 1670 & 2840 & 14400 \\
LBCS & 1.86 & 17.5 & 14.8 & 67.6 & 257 & 353 \\
fermionic shadow & -- & 69.9 & 155 & 586 & 8440 & 5846 \\
Basis rotation grouping & -- & 22.6 & 7.0 & 68.3 & 6559 & 3288 \\
joint fermionic grouping & -- & 106 & 424 & 980 & 5970 & 4240 \\
uniform derandomized shadow & -- & -- & 8.26 & -- & 712 & 833 \\
uniform OGM & -- & -- & 7.62 & -- & 479 & 357 \\
uniform shadow grouping & -- & -- & 6.85 & -- & 432 & 279 \\
C-LBCS & -- & -- & 6.53 & -- & 430 & 287 \\
\end{tabular}
\bigskip
\caption{Variances [Ha$^2$] of several different methods; see Section~\ref{sec:comparison-algorithms} for references.}
\label{tab:var-comparison-existing}
\end{table}

\begin{table}
\begin{tabular}{l|r|r|r|r|r|r}
Method & \ce{H2} (STO-3G) & \ce{H2} (6-31G) & \ce{LiH} & \ce{BeH2} & \ce{H2O} & \ce{NH3} \\ \hline
LDF, QWC + $\ell^2$ & 0.424 & 21.0 & 46.6 & 117 & 921 & 732 \\
LDF, QWC + maxim. + SLSQP & 0.424 & 5.48 & 7.20 & 23.2 & 131 & 172 \\
LDF, FC + $\ell^2$ & 0.352 & 14.8 & 18.4 & 102 & 1070 & 625 \\
LDF, FC + maxim. + SLSQP & \textbf{0.288} & 13.9 & 9.67 & 15.1 & 291 & 145 \\
\hline
LDVF, QWC + $\ell^2$ & 0.424 & 15.8 & 20.2 & 46.6 & 450 & 570 \\
LDVF, QWC + maxim. + SLSQP & 0.424 & 4.91 & 5.64 & 20.9 & 148 & 150 \\
LDVF, FC + $\ell^2$ & 0.352 & 9.12 & 18.5 & 75.9 & 611 & 315 \\
LDVF, FC + maxim. + SLSQP & \textbf{0.288} & 8.32 & 8.01 & 22.6 & 166 & 127 \\
\hline
LVF, QWC + $\ell^2$ & 0.424 & 26.4 & 34.5 & 69.1 & 591 & 503 \\
LVF, QWC + maxim. + SLSQP & 0.424 & 5.69 & 4.98 & 21.1 & 146 & 164 \\
LVF, FC + $\ell^2$ & 0.364 & 5.89 & 23.1 & 74.3 & 430 & 470 \\
LVF, FC + maxim. + SLSQP & \textbf{0.288} & 5.52 & 6.59 & 26.0 & 104 & 108 \\
\hline
RDS + SLSQP & 0.424 & 5.78 & 5.93 & 29.4 & 164 & 197 \\
RDS + QWC maxim. + SLSQP & 0.424 & 5.37 & 5.34 & 26.1 & 144 & 184 \\
RDS + Cliffordization + SLSQP & \textbf{0.288} & \textbf{3.19} & 3.09 & 16.9 & 83.4 & 112 \\
\hline
RGS + SLSQP & 0.424 & 5.67 & 3.67 & 13.8 & 77.1 & 83.8\\
RGS + Cliffordization + SLSQP & \textbf{0.288} & 3.21 & \textbf{2.45} & \textbf{10.1} & \textbf{49.7} & \textbf{50.9}
\end{tabular}
\bigskip
\caption{Variance [Ha$^2$] comparison for Horvitz--Thompson estimators, $\estO_\pi$.}
\label{tab:var-comparison-horvitz-thompson}
\end{table}

\begin{table}
\resizebox{\textwidth}{!}{\begin{tabular}{l|r|r|r|r|r|r}
Method & \ce{H2} (STO-3G) & \ce{H2} (6-31G) & \ce{LiH} & \ce{BeH2} & \ce{H2O} & \ce{NH3} \\ \hline
LDF, QWC & 0.195 & 6.22 & 8.98 & 28.8 & 169 & 333 \\
LDF, QWC + maxim. & 0.195 & 2.29 & 3.83 & 9.34 & 41.0 & 123 \\
LDF, FC & \textbf{0.125} & 1.15 & 2.56 & 6.88 & 93.8 & 187 \\
LDF, FC + maxim. & \textbf{0.125} & 1.24 & 1.42 & \textbf{2.15} & 17.0 & 28.1 \\
\hline
LDVF, QWC & 0.195 & 4.38 & 4.43 & 18.8 & 75.5 & 342 \\
LDVF, QWC + maxim. & 0.195 & 2.31 & 2.71 & 9.58 & 44.4 & 97.7 \\
LDVF, FC & \textbf{0.125} & 2.10 & 3.48 & 6.03 & 58.1 & 81.8 \\
LDVF, FC + maxim. & \textbf{0.125} & 1.71 & 1.42 & 2.50 & 20.7 & 32.8 \\
\hline
LVF, QWC & 0.195 & 4.70 & 7.24 & 31.8 & 225 & 333 \\
LVF, QWC + maxim. & 0.195 & 2.20 & 3.57 & 10.4 & 62.7 & 134 \\
LVF, FC & \textbf{0.125} & 0.956 & 3.99 & 8.07 & 45.9 & 113 \\
LVF, FC + maxim. & \textbf{0.125} & 0.964 & 1.64 & 3.25 & 21.0 & 30.4 \\
\hline
Derandomized shadow & 0.195 & 1.91 & 3.1 & 11.3 & 46.2 & 118.0 \\
Derandomized shadow + QWC maxim. & 0.195 & 1.83 & 2.87 & 9.66 & 42.2 & 107.0 \\
Derandomized shadow + Cliffordization & 0.195 & 0.872 & 1.55 & 6.44 & 23.0 & 54.2 \\
\hline
ShadowGrouping (minimal) & 0.195 & 2.18 & 1.66 & 6.98 & 20.4 & 45.6 \\
ShadowGrouping (minimal) + Cliffordization & 0.195 & 1.12 & 0.933 & 4.76 & 14.7 & 28.8 \\
ShadowGrouping (3$\lvert \mathcal{P} \rvert$) & 0.127 & 1.76 & 1.06 & 3.78 & 11.7 & 23.3 \\
ShadowGrouping (3$\lvert \mathcal{P} \rvert$) + Cliffordization & 0.127 & \textbf{0.737} & \textbf{0.758} & 2.94 & \textbf{8.44} & \textbf{17.9}
\end{tabular}}
\bigskip
\caption{Variance [Ha$^2$] comparison for deterministic estimators with uniform allocation, $\estO^{\mathrm{det},M}_{\pi_{\mathrm{uniform}}}$.}
\label{tab:var-comparison-deterministic-uniform}
\end{table}

\begin{table}
\begin{tabular}{l|r|r|r|r|r|r}
Method & \ce{H2} (STO-3G) & \ce{H2} (6-31G) & \ce{LiH} & \ce{BeH2} & \ce{H2O} & \ce{NH3} \\ \hline
LDF, QWC & 0.157 & 8.59 & 12.7 & 38.9 & 232 & 295 \\
LDF, QWC + maxim. & 0.157 & 2.62 & 2.86 & 10.8 & 29.3 & 72.3 \\
LDF, FC & 0.259 & 1.84 & 3.95 & 10.3 & 183 & 163 \\
LDF, FC + maxim. & 0.128 & 1.36 & 0.95 & \textbf{2.31} & 14.7 & 30.1 \\
\hline
LDVF, QWC & 0.157 & 5.36 & 3.95 & 13.1 & 94.0 & 153 \\
LDVF, QWC + maxim. & 0.157 & 2.61 & 1.97 & 9.67 & 25.5 & 63.0 \\
LDVF, FC & 0.259 & 4.00 & 4.02 & 9.44 & 94.9 & 83.7 \\
LDVF, FC + maxim. & 0.128 & 2.13 & 1.61 & 2.74 & 17.0 & 61.4 \\
\hline
LVF, QWC & 0.157 & 5.40 & 6.18 & 24.8 & 174 & 187 \\
LVF, QWC + maxim. & 0.157 & 3.26 & 2.45 & 7.27 & 30.6 & 73.1 \\
LVF, FC & \textbf{0.125} & 1.28 & 3.60 & 10.7 & 76.0 & 114.0 \\
LVF, FC + maxim. & 0.128 & 0.954 & 1.68 & 3.16 & 12.8 & 24.1 \\
\hline
ShadowGrouping (3$\lvert \mathcal{P} \rvert$) & 0.157 & 2.72 & 1.66 & 4.79 & 13.3 & 28.3 \\
ShadowGrouping (3$\lvert \mathcal{P} \rvert$) + Cliffordization & 0.128 & \textbf{0.836} & \textbf{0.703} & 3.04 & \textbf{8.19} & \textbf{15.0}
\end{tabular}
\bigskip
\caption{Variance [Ha$^2$] comparison for deterministic estimators with $\ell^2$ allocation, $\estO^{\mathrm{det},M}_{\pi_{\ell^2}}$.}
\label{tab:var-comparison-deterministic-l2}
\end{table}

By taking the lowest obtained variance for each benchmark molecule, we may note the lowest possible number of measurements required to estimate its ground state energy to chemical accuracy, i.e. with an error of at most $0.0016 \, \text{Ha}$, with probability $95\,\%$, according to the CLT approximation (Section~\ref{sec:clt}); see Table~\ref{tab:sota}.

\begin{table}
    \makegapedcells
    \centering
    \begin{tabular}{l||r|r|l}
        System & Number of measurements & Variance [Ha$^2$] & Estimator\\
        \hline
        \ce{H2} (STO-3G) &187,572&0.125&\makecell[l]{Several}\\
        \hline
        \ce{H2} (6-31G) &1,105,920&0.737&\makecell[l]{ShadowGrouping ($3\lvert \calP \rvert$)\\Cliffordization \\Uniform allocation}\\
        \hline
        \ce{LiH} &1,054,901&0.703& \makecell[l]{ShadowGrouping ($3\lvert \calP \rvert$) \\Cliffordization\\ $\ell^2$ allocation}\\
        \hline
        \ce{BeH2} &3,226,226&2.15& \makecell[l]{LDF, FC \\ Maximalization \\ Uniform allocation}\\ 
        \hline
        \ce{H2O} &12,289,668&8.19&\makecell[l]{ShadowGrouping ($3\lvert \calP \rvert$) \\ Cliffordization \\ $\ell^2$ allocation}\\ 
        \hline
        \ce{NH3} &22,508,548&15.0&\makecell[l]{ShadowGrouping ($3\lvert \calP \rvert$) \\ Cliffordization \\ $\ell^2$ allocation}\\
    \end{tabular}
    \caption{Out of all the ones considered, the lowest variance estimator for each of the benchmark molecules, and the number of measurements required to reach chemical accuracy with probability 95\,\%.}
    \label{tab:sota}
\end{table}

\section{Code availability}
\label{sec:codeavailability}
See \url{https://codeberg.org/kvantify/paper-variance-reduction} for the source code used to generate  the results in this paper.
\clearpage

\bibliographystyle{is-alpha}
\bibliography{references}

@ARTICLE{Aaronson2020-tf,
  title     = "Shadow tomography of quantum states",
  author    = "Aaronson, Scott",
  journal   = "SIAM J. Comput.",
  publisher = "Society for Industrial \& Applied Mathematics (SIAM)",
  volume    =  49,
  number    =  5,
  pages     = "STOC18--368--STOC18--394",
  month     =  jan,
  year      =  2020,
  language  = "en"
}

@ARTICLE{Aaronson2004-ff,
  title     = "Improved simulation of stabilizer circuits",
  author    = "Aaronson, Scott and Gottesman, Daniel",
  journal   = "Phys. Rev. A",
  publisher = "American Physical Society (APS)",
  volume    =  70,
  number    =  5,
  month     =  nov,
  year      =  2004,
  copyright = "http://link.aps.org/licenses/aps-default-license",
  language  = "en"
}

@inproceedings{bb-maxclique,
    booktitle = {Simpósio Brasileiro de Pesquisa Operacional},
    author = {Tavares, Wladimir and Campêlo, Manoel and Rodrigues, Diego and Michelon, Philippe},
    year = {2015},
    month = aug,
    pages = {},
    title = {Um Algoritmo de Branch and Bound para o Problema da Clique Máxima Ponderada},
    publisher = {\ }
}

@ARTICLE{Bian2025-au,
  title     = "Adaptive-depth randomized measurement for fermionic observables",
  author    = "Bian, Kaiming and Wu, Bujiao",
  journal   = "Quantum Sci. Technol.",
  publisher = "IOP Publishing",
  volume    =  10,
  number    =  3,
  pages     = "035063",
  month     =  oct,
  year      =  2025,
  copyright = "https://iopscience.iop.org/page/copyright"
}

@ARTICLE{Bonet-Monroig2020-qg,
  title     = "Nearly optimal measurement scheduling for partial tomography of
               quantum states",
  author    = "Bonet-Monroig, Xavier and Babbush, Ryan and O'Brien, Thomas E.",
  journal   = "Phys. Rev. X.",
  publisher = "American Physical Society (APS)",
  volume    =  10,
  number    =  3,
  month     =  sep,
  year      =  2020,
  copyright = "https://creativecommons.org/licenses/by/4.0/",
  language  = "en"
}

@ARTICLE{Burns2025-do,
  title     = "{GALIC}: hybrid multi-qubitwise pauli grouping for quantum
               computing measurement",
  author    = "Burns, Matthew X. and Liu, Chenxu and Stein, Samuel and Peng, Bo
               and Kowalski, Karol and Li, Ang",
  journal   = "Quantum Sci. Technol.",
  publisher = "IOP Publishing",
  volume    =  10,
  number    =  1,
  pages     = "015054",
  month     =  jan,
  year      =  2025,
  copyright = "https://creativecommons.org/licenses/by/4.0/"
}

@ARTICLE{Crawford2021-gz,
  title     = "{Efficient quantum measurement of Pauli operators in the presence
               of finite sampling error}",
  author    = "Crawford, Ophelia and van Straaten, Barnaby and Wang, Daochen
               and Parks, Thomas and Campbell, Earl and Brierley, Stephen",
  journal   = "Quantum",
  publisher = "Verein zur Forderung des Open Access Publizierens in den
               Quantenwissenschaften",
  volume    =  5,
  number    =  385,
  pages     = "385",
  month     =  jan,
  year      =  2021,
  language  = "en"
}

@ARTICLE{Elben2022-wz,
  title     = "The randomized measurement toolbox",
  author    = "Elben, Andreas and Flammia, Steven T. and Huang, Hsin-Yuan and
               Kueng, Richard and Preskill, John and Vermersch, Beno{\^\i}t and
               Zoller, Peter",
  journal   = "Nat. Rev. Phys.",
  publisher = "Springer Science and Business Media LLC",
  volume    =  5,
  number    =  1,
  pages     = "9--24",
  month     =  dec,
  year      =  2022,
  copyright = "https://www.springernature.com/gp/researchers/text-and-data-mining",
  language  = "en"
}

@ARTICLE{Foster1976-oa,
  title     = "An integer programming approach to the vehicle scheduling
               problem",
  author    = "Foster, Brian A. and Ryan, David M.",
  journal   = "Oper. Res. Q.",
  publisher = "JSTOR",
  volume    =  27,
  number    =  2,
  pages     = "367",
  year      =  1976
}

@article{gidney2021stim,
  doi = {10.22331/q-2021-07-06-497},
  title = {Stim: a fast stabilizer circuit simulator},
  author = {Gidney, Craig},
  journal = {{Quantum}},
  issn = {2521-327X},
  publisher = {{Verein zur F{\"{o}}rderung des Open Access Publizierens
                in den Quantenwissenschaften}},
  volume = 5,
  pages = 497,
  month = jul,
  year = 2021
}

@ARTICLE{Gokhale2020-jc,
  title     = {{$O(N^3)$} measurement cost for variational quantum eigensolver
               on molecular {H}amiltonians},
  author    = "Gokhale, Pranav and Angiuli, Olivia and Ding, Yongshan and Gui,
               Kaiwen and Tomesh, Teague and Suchara, Martin and Martonosi,
               Margaret and Chong, Frederic T.",
  journal   = "IEEE Trans. Quantum Eng.",
  publisher = "Institute of Electrical and Electronics Engineers (IEEE)",
  volume    =  1,
  pages     = "1--24",
  year      =  2020,
  copyright = "https://creativecommons.org/licenses/by/4.0/legalcode"
}

@ARTICLE{Gresch2025-kd,
  title     = "Reducing the sampling complexity of energy estimation in quantum
               many-body systems using empirical variance information",
  author    = "Gresch, Alexander and Tepe, U{\u g}ur and Kliesch, Martin",
  journal   = "J. Chem. Theory Comput.",
  publisher = "American Chemical Society (ACS)",
  volume    =  21,
  number    =  15,
  pages     = "7352--7359",
  month     =  aug,
  year      =  2025,
  language  = "en"
}

@ARTICLE{Gresch2025-zy,
  title     = "Guaranteed efficient energy estimation of quantum many-body
               {Hamiltonians} using {ShadowGrouping}",
  author    = "Gresch, Alexander and Kliesch, Martin",
  journal   = "Nat. Commun.",
  publisher = "Springer Science and Business Media LLC",
  volume    =  16,
  number    =  1,
  pages     = "689",
  month     =  jan,
  year      =  2025,
  copyright = "https://creativecommons.org/licenses/by/4.0",
  language  = "en"
}

@ARTICLE{Grimsley2019-cp,
  title     = "An adaptive variational algorithm for exact molecular
               simulations on a quantum computer",
  author    = "Grimsley, Harper R. and Economou, Sophia E. and Barnes, Edwin and
               Mayhall, Nicholas J.",
  journal   = "Nat. Commun.",
  publisher = "Springer Science and Business Media LLC",
  volume    =  10,
  number    =  1,
  pages     = "3007",
  month     =  jul,
  year      =  2019,
  copyright = "https://creativecommons.org/licenses/by/4.0",
  language  = "en"
}

@ARTICLE{Heyraud2025-zs,
  title    = "Unified framework for matchgate classical shadows",
  author   = "Heyraud, Valentin and Chomet, H{\'e}loise and Tilly, Jules",
  journal  = "Npj Quantum Inf.",
  volume   =  11,
  number   =  1,
  pages    = "65",
  month    =  apr,
  year     =  2025,
  language = "en"
}

@ARTICLE{Huang2020-ws,
  title     = "Predicting many properties of a quantum system from very few
               measurements",
  author    = "Huang, Hsin-Yuan and Kueng, Richard and Preskill, John",
  journal   = "Nat. Phys.",
  publisher = "Springer Science and Business Media LLC",
  volume    =  16,
  number    =  10,
  pages     = "1050--1057",
  month     =  oct,
  year      =  2020,
  language  = "en"
}

@ARTICLE{Huang2021-ga,
  title     = "{Efficient estimation of Pauli observables by derandomization}",
  author    = "Huang, Hsin-Yuan and Kueng, Richard and Preskill, John",
  journal   = "Phys. Rev. Lett.",
  publisher = "American Physical Society (APS)",
  volume    =  127,
  number    =  3,
  pages     = "030503",
  month     =  jul,
  year      =  2021,
  copyright = "https://link.aps.org/licenses/aps-default-license",
  language  = "en"
}

@article{Hadfield2021-hx,
  title        = "{Adaptive Pauli Shadows for energy estimation}",
  author       = "Hadfield, Charles",
  year         =  2021,
  primaryClass = "quant-ph",
  eprint       = "2105.12207",
  doi = "10.48550/arXiv.2105.12207",
  journal={arXiv preprint arXiv:2105.12207}
}

@ARTICLE{Hadfield2022-vc,
  title     = "{Measurements of quantum Hamiltonians with locally-biased
               classical shadows}",
  author    = "Hadfield, Charles and Bravyi, Sergey and Raymond, Rudy and
               Mezzacapo, Antonio",
  journal   = "Commun. Math. Phys.",
  publisher = "Springer Science and Business Media LLC",
  volume    =  391,
  number    =  3,
  pages     = "951--967",
  month     =  may,
  year      =  2022,
  copyright = "https://creativecommons.org/licenses/by/4.0",
  language  = "en"
}

@ARTICLE{Huggins2021-al,
  title     = "Efficient and noise resilient measurements for quantum chemistry
               on near-term quantum computers",
  author    = "Huggins, William J. and McClean, Jarrod R. and Rubin, Nicholas C.
               and Jiang, Zhang and Wiebe, Nathan and Whaley, K. Birgitta and
               Babbush, Ryan",
  journal   = "Npj Quantum Inf.",
  publisher = "Springer Science and Business Media LLC",
  volume    =  7,
  number    =  1,
  month     =  feb,
  year      =  2021,
  copyright = "https://creativecommons.org/licenses/by/4.0",
  language  = "en"
}

@ARTICLE{Huangfu2018-sw,
  title     = "Parallelizing the dual revised simplex method",
  author    = "Huangfu, Qi and Hall, Julian A. J.",
  journal   = "Math. Program. Comput.",
  publisher = "Springer Nature",
  volume    =  10,
  number    =  1,
  pages     = "119--142",
  month     =  mar,
  year      =  2018
}

@ARTICLE{Izmaylov2019-gd,
  title     = "Revising the measurement process in the variational quantum
               eigensolver: is it possible to reduce the number of separately
               measured operators?",
  author    = "Izmaylov, Artur F. and Yen, Tzu-Ching and Ryabinkin, Ilya G.",
  journal   = "Chem. Sci.",
  publisher = "Royal Society of Chemistry (RSC)",
  volume    =  10,
  number    =  13,
  pages     = "3746--3755",
  month     =  apr,
  year      =  2019,
  copyright = "http://creativecommons.org/licenses/by/3.0/",
  language  = "en"
}

@ARTICLE{Jiang2020-ls,
  title     = "Optimal fermion-to-qubit mapping via ternary trees with
               applications to reduced quantum states learning",
  author    = "Jiang, Zhang and Kalev, Amir and Mruczkiewicz, Wojciech and
               Neven, Hartmut",
  journal   = "Quantum",
  publisher = "Verein zur Forderung des Open Access Publizierens in den
               Quantenwissenschaften",
  volume    =  4,
  number    =  276,
  pages     = "276",
  month     =  jun,
  year      =  2020,
  language  = "en"
}

@article{Jena2019-hx,
  title        = "{Pauli Partitioning with Respect to Gate Sets}",
  author       = "Jena, Andrew and Genin, Scott and Mosca, Michele",
  year         =  2019,
  primaryClass = "quant-ph",
  eprint       = "1907.07859",
  doi = "10.48550/arXiv.1907.07859",
  journal={arXiv preprint arXiv:1907.07859},
}

@ARTICLE{Kraft1988,
  title     = "A software package for sequential quadratic programming",
  author    = "Kraft, Dieter",
  journal   = "Tech. Rep. DFVLR-FB 88-28",
  publisher = "DLR German Aerospace Center — Institute for Flight Mechanics",
  year      =  1988
}

@BOOK{Lawson1995-gl,
  title     = "Solving least squares problems",
  author    = "Lawson, Charles L.",
  publisher = "SIAM",
  series    = "Classics in applied mathematics",
  year      =  1995,
  address   = "Philadelphia, MS",
  language  = "en"
}

@ARTICLE{Leighton1979-zw,
  title     = "A graph coloring algorithm for large scheduling problems",
  author    = "Leighton, Frank Thomson",
  journal   = "J. Res. Natl. Bur. Stand. (1977)",
  publisher = "National Institute of Standards and Technology (NIST)",
  volume    =  84,
  number    =  6,
  pages     = "489--506",
  month     =  nov,
  year      =  1979,
  language  = "en"
}

@ARTICLE{Li2025-vm,
  title     = "{Resource-Optimized} {Grouping Shadow} for efficient energy
               estimation",
  author    = "Li, Min and Lin, Mao and Beach, Matthew J. S.",
  journal   = "Quantum",
  publisher = "Verein zur Forderung des Open Access Publizierens in den
               Quantenwissenschaften",
  volume    =  9,
  number    =  1694,
  pages     = "1694",
  month     =  apr,
  year      =  2025,
  language  = "en"
}

@ARTICLE{Liang2024-qd,
  title     = "Artificial-intelligence-driven shot reduction in quantum
               measurement",
  author    = "Liang, Senwei and Zhu, Linghua and Liu, Xiaolin and Yang, Chao
               and Li, Xiaosong",
  journal   = "Chem. Phys. Rev.",
  publisher = "AIP Publishing",
  volume    =  5,
  number    =  4,
  month     =  dec,
  year      =  2024,
  language  = "en"
}

@ARTICLE{Majland2023-lz,
  title     = "Fermionic adaptive sampling theory for variational quantum
               eigensolvers",
  author    = "Majland, Marco and Ettenhuber, Patrick and Zinner, Nikolaj
               Thomas",
  journal   = "Phys. Rev. A (Coll. Park.)",
  publisher = "American Physical Society (APS)",
  volume    =  108,
  number    =  5,
  month     =  nov,
  year      =  2023,
  copyright = "https://link.aps.org/licenses/aps-default-license",
  language  = "en"
}

@ARTICLE{Majsak2025-av,
  title     = "A simple and efficient joint measurement strategy for estimating
               fermionic observables and {H}amiltonians",
  author    = "Majsak, Joanna and McNulty, Daniel and Oszmaniec, Micha{\l}",
  journal   = "Npj Quantum Inf.",
  publisher = "Springer Science and Business Media LLC",
  volume    =  11,
  number    =  1,
  month     =  apr,
  year      =  2025,
  copyright = "https://creativecommons.org/licenses/by-nc-nd/4.0",
  language  = "en"
}

@ARTICLE{McClean2016-er,
  title     = "The theory of variational hybrid quantum-classical algorithms",
  author    = "McClean, Jarrod R. and Romero, Jonathan and Babbush, Ryan and
               Aspuru-Guzik, Al{\'a}n",
  journal   = "New J. Phys.",
  publisher = "IOP Publishing",
  volume    =  18,
  number    =  2,
  pages     = "023023",
  month     =  feb,
  year      =  2016,
  copyright = "http://creativecommons.org/licenses/by/3.0/"
}

@InProceedings{networkx,
  author =       {Aric A. Hagberg and Daniel A. Schult and Pieter J. Swart},
  title =        {{Exploring Network Structure, Dynamics, and Function using NetworkX}},
  booktitle =   {Proceedings of the 7th Python in Science Conference},
  pages =     {11 - 15},
  address = {Pasadena, CA USA},
  publisher = {SciPy},
  year =      {2008},
  editor =    {Ga\"el Varoquaux and Travis Vaught and Jarrod Millman},
}

@ARTICLE{Peruzzo2014-uw,
  title     = "A variational eigenvalue solver on a photonic quantum processor",
  author    = "Peruzzo, Alberto and McClean, Jarrod and Shadbolt, Peter and
               Yung, Man-Hong and Zhou, Xiao-Qi and Love, Peter J. and
               Aspuru-Guzik, Al{\'a}n and O'Brien, Jeremy L.",
  journal   = "Nat. Commun.",
  publisher = "Springer Science and Business Media LLC",
  volume    =  5,
  number    =  1,
  pages     = "4213",
  month     =  jul,
  year      =  2014,
  copyright = "https://creativecommons.org/licenses/by-nc-nd/4.0",
  language  = "en"
}

@ARTICLE{SciPy2020-NMeth,
  author  = {Virtanen, Pauli and Gommers, Ralf and Oliphant, Travis E. and
            Haberland, Matt and Reddy, Tyler and Cournapeau, David and
            Burovski, Evgeni and Peterson, Pearu and Weckesser, Warren and
            Bright, Jonathan and {van der Walt}, St{\'e}fan J. and
            Brett, Matthew and Wilson, Joshua and Millman, K. Jarrod and
            Mayorov, Nikolay and Nelson, Andrew R. J. and Jones, Eric and
            Kern, Robert and Larson, Eric and Carey, C J and
            Polat, {\.I}lhan and Feng, Yu and Moore, Eric W. and
            {VanderPlas}, Jake and Laxalde, Denis and Perktold, Josef and
            Cimrman, Robert and Henriksen, Ian and Quintero, E. A. and
            Harris, Charles R. and Archibald, Anne M. and
            Ribeiro, Ant{\^o}nio H. and Pedregosa, Fabian and
            {van Mulbregt}, Paul and {SciPy 1.0 Contributors}},
  title   = {{{SciPy} 1.0: Fundamental Algorithms for Scientific
            Computing in Python}},
  journal = {Nature Methods},
  year    = {2020},
  volume  = {17},
  pages   = {261--272},
  doi     = {10.1038/s41592-019-0686-2},
}

@ARTICLE{Shlosberg2023-ok,
  title     = "Adaptive estimation of quantum observables",
  author    = "Shlosberg, Ariel and Jena, Andrew J. and Mukhopadhyay, Priyanka
               and Haase, Jan F. and Leditzky, Felix and Dellantonio, Luca",
  journal   = "Quantum",
  publisher = "Verein zur Forderung des Open Access Publizierens in den
               Quantenwissenschaften",
  volume    =  7,
  number    =  906,
  pages     = "906",
  month     =  jan,
  year      =  2023,
  language  = "en"
}

@article{Sivarajah_TKET_A_Retargetable_2020,author = {Sivarajah, Seyon and Dilkes, Silas and Cowtan, Alexander and Simmons, Will and Edgington, Alec and Duncan, Ross},doi = {10.1088/2058-9565/ab8e92},journal = {Quantum Science and Technology},month = nov,title = {{TKET: A Retargetable Compiler for NISQ devices}},volume = {6},year = {2020}}

@ARTICLE{Verteletskyi2020-qz,
  title     = "Measurement optimization in the variational quantum eigensolver
               using a minimum clique cover",
  author    = "Verteletskyi, Vladyslav and Yen, Tzu-Ching and Izmaylov, Artur F.",
  journal   = "J. Chem. Phys.",
  publisher = "AIP Publishing",
  volume    =  152,
  number    =  12,
  pages     = "124114",
  month     =  mar,
  year      =  2020,
  language  = "en"
}

@ARTICLE{Wan2023-yf,
  title     = "Matchgate shadows for fermionic quantum simulation",
  author    = "Wan, Kianna and Huggins, William J. and Lee, Joonho and Babbush,
               Ryan",
  journal   = "Commun. Math. Phys.",
  publisher = "Springer Science and Business Media LLC",
  volume    =  404,
  number    =  2,
  pages     = "629--700",
  month     =  dec,
  year      =  2023,
  copyright = "https://creativecommons.org/licenses/by/4.0",
  language  = "en"
}

@ARTICLE{Wu2023-sw,
  title     = "{Overlapped grouping measurement: A unified framework for
               measuring quantum states}",
  author    = "Wu, Bujiao and Sun, Jinzhao and Huang, Qi and Yuan, Xiao",
  journal   = "Quantum",
  publisher = "Verein zur Forderung des Open Access Publizierens in den
               Quantenwissenschaften",
  volume    =  7,
  number    =  896,
  pages     = "896",
  month     =  jan,
  year      =  2023,
  language  = "en"
}

@ARTICLE{Yen2020-lk,
  title     = "Measuring all compatible operators in one series of single-qubit
               measurements using unitary transformations",
  author    = "Yen, Tzu-Ching and Verteletskyi, Vladyslav and Izmaylov, Artur F.",
  journal   = "J. Chem. Theory Comput.",
  publisher = "American Chemical Society (ACS)",
  volume    =  16,
  number    =  4,
  pages     = "2400--2409",
  month     =  apr,
  year      =  2020,
  language  = "en"
}

@ARTICLE{Yen2023-zv,
  title    = "Deterministic improvements of quantum measurements with grouping
              of compatible operators, non-local transformations, and
              covariance estimates",
  author   = "Yen, Tzu-Ching and Ganeshram, Aadithya and Izmaylov, Artur F.",
  journal  = "NPJ Quantum Inf.",
  volume   =  9,
  number   =  1,
  pages    = "14",
  month    =  feb,
  year     =  2023,
  language = "en"
}

@misc{Zhang2023-lr,
  title        = "A composite measurement scheme for efficient quantum
                  observable estimation",
  author       = "Zhang, Zi-Jian and Nakaji, Kouhei and Choi, Matthew and
                  Aspuru-Guzik, Al{\'a}n",
  year         =  2023,
  primaryClass = "quant-ph",
  eprint       = "2305.02439",
  journal={arXiv preprint arXiv:2305.02439},
}

@ARTICLE{Zhao2020-xa,
  title     = "Measurement reduction in variational quantum algorithms",
  author    = "Zhao, Andrew and Tranter, Andrew and Kirby, William M. and Ung,
               Shu Fay and Miyake, Akimasa and Love, Peter J.",
  journal   = "Phys. Rev. A (Coll. Park.)",
  publisher = "American Physical Society (APS)",
  volume    =  101,
  number    =  6,
  month     =  jun,
  year      =  2020,
  copyright = "https://link.aps.org/licenses/aps-default-license",
  language  = "en"
}

@ARTICLE{Zhao2021-gn,
  title     = "Fermionic partial tomography via classical shadows",
  author    = "Zhao, Andrew and Rubin, Nicholas C. and Miyake, Akimasa",
  journal   = "Phys. Rev. Lett.",
  publisher = "American Physical Society (APS)",
  volume    =  127,
  number    =  11,
  pages     = "110504",
  month     =  sep,
  year      =  2021,
  copyright = "https://link.aps.org/licenses/aps-default-license",
  language  = "en"
}

@ARTICLE{Zhu2024-na,
  title     = "Active learning for quantum mechanical measurements",
  author    = "Zhu, Ruidi and Pike-Burke, Ciara and Mintert, Florian",
  journal   = "Phys. Rev. A (Coll. Park.)",
  publisher = "American Physical Society (APS)",
  volume    =  109,
  number    =  6,
  month     =  jun,
  year      =  2024,
  copyright = "https://creativecommons.org/licenses/by/4.0/",
  language  = "en"
}

@article{tilly2022variational,
  title={The variational quantum eigensolver: a review of methods and best practices},
  author={Tilly, Jules and Chen, Hongxiang and Cao, Shuxiang and Picozzi, Dario and Setia, Kanav and Li, Ying and Grant, Edward and Wossnig, Leonard and Rungger, Ivan and Booth, George H and others},
  journal={Physics Reports},
  volume={986},
  pages={1--128},
  year={2022},
  publisher={Elsevier}
}

@article{bharti2022noisy,
  title={Noisy intermediate-scale quantum algorithms},
  author={Bharti, Kishor and Cervera-Lierta, Alba and Kyaw, Thi Ha and Haug, Tobias and Alperin-Lea, Sumner and Anand, Abhinav and Degroote, Matthias and Heimonen, Hermanni and Kottmann, Jakob S and Menke, Tim and others},
  journal={Reviews of Modern Physics},
  volume={94},
  number={1},
  pages={015004},
  year={2022},
  publisher={APS}
}

@article{cerezo2021variational,
  title={Variational quantum algorithms},
  author={Cerezo, Marco and Arrasmith, Andrew and Babbush, Ryan and Benjamin, Simon C. and Endo, Suguru and Fujii, Keisuke and McClean, Jarrod R. and Mitarai, Kosuke and Yuan, Xiao and Cincio, Lukasz and others},
  journal={Nature Reviews Physics},
  volume={3},
  number={9},
  pages={625--644},
  year={2021},
  publisher={Nature Publishing Group UK London}
}

@article{chan2025algorithmic,
  title={Algorithmic shadow spectroscopy},
  author={Chan, Hans Hon Sang and Meister, Richard and Goh, Matthew L. and Koczor, B{\'a}lint},
  journal={PRX Quantum},
  volume={6},
  number={1},
  pages={010352},
  year={2025},
  publisher={APS}
}

@article{sack2022avoiding,
  title={Avoiding barren plateaus using classical shadows},
  author={Sack, Stefan H. and Medina, Raimel A. and Michailidis, Alexios A. and Kueng, Richard and Serbyn, Maksym},
  journal={PRX Quantum},
  volume={3},
  number={2},
  pages={020365},
  year={2022},
  publisher={APS}
}

@article{boyd2022training,
  title={{Training variational quantum circuits with CoVaR: Covariance root finding with classical shadows}},
  author={Boyd, Gregory and Koczor, B{\'a}lint},
  journal={Physical Review X},
  volume={12},
  number={4},
  pages={041022},
  year={2022},
  publisher={APS}
}

@article{chen2021robust,
  title={Robust shadow estimation},
  author={Chen, Senrui and Yu, Wenjun and Zeng, Pei and Flammia, Steven T.},
  journal={PRX Quantum},
  volume={2},
  number={3},
  pages={030348},
  year={2021},
  publisher={APS}
}

@article{karaiskos2025hard,
  title={How hard is it to verify a classical shadow?},
  author={Karaiskos, Georgios and Rudolph, Dorian and Meyer, Johannes Jakob and Eisert, Jens and Gharibian, Sevag},
  journal={arXiv preprint arXiv:2510.08515},
  year={2025}
}

@article{levy2024classical,
  title={Classical shadows for quantum process tomography on near-term quantum computers},
  author={Levy, Ryan and Luo, Di and Clark, Bryan K.},
  journal={Physical Review Research},
  volume={6},
  number={1},
  pages={013029},
  year={2024},
  publisher={APS}
}

@article{arrasmith2023development,
  title={{Development and demonstration of an efficient readout error mitigation technique for use in NISQ algorithms}},
  author={Arrasmith, Andrew and Patterson, Andrew and Boughton, Alice and Paini, Marco},
  journal={arXiv preprint arXiv:2303.17741},
  year={2023}
}

@article{zhao2024group,
  title={Group-theoretic error mitigation enabled by classical shadows and symmetries},
  author={Zhao, Andrew and Miyake, Akimasa},
  journal={npj Quantum Information},
  volume={10},
  number={1},
  pages={57},
  year={2024},
  publisher={Nature Publishing Group UK London}
}

@article{seif2023shadow,
  title={{Shadow distillation: Quantum error mitigation with classical shadows for near-term quantum processors}},
  author={Seif, Alireza and Cian, Ze-Pei and Zhou, Sisi and Chen, Senrui and Jiang, Liang},
  journal={PRX Quantum},
  volume={4},
  number={1},
  pages={010303},
  year={2023},
  publisher={APS}
}

@article{wu2024error,
  title={Error-mitigated fermionic classical shadows on noisy quantum devices},
  author={Wu, Bujiao and Koh, Dax Enshan},
  journal={npj Quantum Information},
  volume={10},
  number={1},
  pages={39},
  year={2024},
  publisher={Nature Publishing Group UK London}
}

@article{jnane2024quantum,
  title={Quantum error mitigated classical shadows},
  author={Jnane, Hamza and Steinberg, Jonathan and Cai, Zhenyu and Nguyen, H. Chau and Koczor, B{\'a}lint},
  journal={PRX Quantum},
  volume={5},
  number={1},
  pages={010324},
  year={2024},
  publisher={APS}
}

@article{amsler2023quantum,
  title={{Quantum-enhanced quantum Monte Carlo: an industrial view}},
  author={Amsler, Maximilian and Deglmann, Peter and Degroote, Matthias and Kaicher, Michael P. and Kiser, Matthew and K{\"u}hn, Michael and Kumar, Chandan and Maier, Andreas and Samsonidze, Georgy and Schroeder, Anna and others},
  journal={arXiv preprint arXiv:2301.11838},
  year={2023}
}

@article{huggins2022unbiasing,
  title={{Unbiasing fermionic quantum Monte Carlo with a quantum computer}},
  author={Huggins, William J. and O’Gorman, Bryan A. and Rubin, Nicholas C. and Reichman, David R. and Babbush, Ryan and Lee, Joonho},
  journal={Nature},
  volume={603},
  number={7901},
  pages={416--420},
  year={2022},
  publisher={Nature Publishing Group UK London}
}

@article{huang2024evaluating,
  title={{Evaluating a quantum-classical quantum Monte Carlo algorithm with Matchgate shadows}},
  author={Huang, Benchen and Chen, Yi-Ting and Gupt, Brajesh and Suchara, Martin and Tran, Anh and McArdle, Sam and Galli, Giulia},
  journal={Physical Review Research},
  volume={6},
  number={4},
  pages={043063},
  year={2024},
  publisher={APS}
}

@article{kiser2024classical,
  title={{Classical and quantum cost of measurement strategies for quantum-enhanced auxiliary field quantum Monte Carlo}},
  author={Kiser, Matthew and Schroeder, Anna and Anselmetti, Gian-Luca R. and Kumar, Chandan and Moll, Nikolaj and Streif, Michael and Vodola, Davide},
  journal={New Journal of Physics},
  volume={26},
  number={3},
  pages={033022},
  year={2024},
  publisher={IOP Publishing}
}

@article{zhao2025quantum,
  title={{Quantum-classical auxiliary field quantum Monte Carlo with matchgate shadows on trapped ion quantum computers}},
  author={Zhao, Luning and Goings, Joshua J. and Aboumrad, Willie and Arrasmith, Andrew and Calderin, Lazaro and Churchill, Spencer and Gabay, Dor and Harvey-Brown, Thea and Hiles, Melanie and Kaja, Magda and others},
  journal={arXiv preprint arXiv:2506.22408},
  year={2025}
}

@article{thomas2025shedding,
  title={Shedding light on classical shadows: learning photonic quantum states},
  author={Thomas, Hugo and Chabaud, Ulysse and Emeriau, Pierre-Emmanuel},
  journal={arXiv preprint arXiv:2510.07240},
  year={2025}
}

@article{birke2026demonstrating,
  title={{Demonstrating and Benchmarking Classical Shadows for Lindblad Tomography}},
  author={Birke, Rune Thinggaard and Severin, Johann Bock and Marciniak, Malthe A. and Hogedal, Emil and Nylander, Andreas and Ahmad, Irshad and Osman, Amr and Bizn{\'a}rov{\'a}, Janka and Rommel, Marcus and Roudsari, Anita Fadavi and others},
  journal={arXiv preprint arXiv:2602.14694},
  year={2026}
}

@article{huang2022quantum,
  title={Quantum advantage in learning from experiments},
  author={Huang, Hsin-Yuan and Broughton, Michael and Cotler, Jordan and Chen, Sitan and Li, Jerry and Mohseni, Masoud and Neven, Hartmut and Babbush, Ryan and Kueng, Richard and Preskill, John and others},
  journal={Science},
  volume={376},
  number={6598},
  pages={1182--1186},
  year={2022},
  publisher={American Association for the Advancement of Science}
}

\end{document}